\documentclass{article}
\usepackage{amsthm}
\theoremstyle{definition}
\newtheorem{definition}{Definition}

\theoremstyle{theorem}
\newtheorem{proposition}{Proposition}
\newtheorem{lemma}{Lemma}
\newtheorem{theorem}{Theorem}
\newtheorem{corollary}{Corollary}
\theoremstyle{definition}

\newtheorem{fact}{Fact}
\pdfoutput=1

\newcommand{\cf}{cf.\ }
\newcommand{\st}{s.t.\ }

\usepackage{xspace}
\newcommand{\etc}{etc.\xspace}
\newcommand{\eg}{e.g.,\xspace}
\newcommand{\ie}{i.e.,\xspace}

\usepackage{array}
\usepackage[pdftex,dvipsnames]{color}
\usepackage{ifthen,calc,url}

\usepackage{amstext,amsmath,amssymb,stmaryrd}
\usepackage[all,cmtip]{xy}

\newcommand{\defeq}{\mathrel{:=}}
\newcommand{\defiff}{\mbox{:iff}}
\newcommand{\bnfeq}{\mathrel{::=}}
\newcommand{\bnfor}{\ \big| \ }

\newcommand{\set}[1]{\{#1\}}
\newcommand{\setst}[2]{\{\ #1\ \boldsymbol{|}\ #2\ \}}

\newcommand{\powerset}[1]{2^{#1}}%{\mathrm{p}}
\newcommand{\nn}[1]{\makebox[2em][r]{#1}\ \ }
\newcommand{\true}{\text{true}}
\newcommand{\false}{\text{false}}

\newcommand{\Scomb}{\mathtt{S}}
\newcommand{\Kcomb}{\mathtt{K}}

\newcommand{\agents}{\mathcal{A}}
\newcommand{\community}{\mathcal{C}}

\newcommand{\messages}{\mathcal{M}}
\newcommand{\data}{\mathcal{D}}

\newcommand{\pair}[2]{(#1,#2)}
\newcommand{\sign}[2]{{\{\negmedspace[#1]\negmedspace\}}_{#2}}

\newcommand{\derives}[3]{#3\vdash_{#1}#2}

\newcommand{\states}{\mathcal{S}}

\newcommand{\msgs}[1]{\mathrm{msgs}_{#1}}
\newcommand{\preorder}[1]{\leq_{#1}}
\newcommand{\indist}[3]{#2\equiv_{#1}#3}

\newcommand{\pAccess}[3]{\mathrel{_{#1}\negthinspace\mathrm{R}_{#2}^{#3}}}
\newcommand{\access}[3]{\mathrel{_{#1}\negthinspace\mathcal{R}_{#2}^{#3}}}

\newcommand{\clo}[2]{\mathrm{cl}_{#1}^{#2}}
\newcommand{\Clo}[2]{\mathrm{Cl}_{#1}^{#2}}
\newcommand{\LiP}{\mathrm{LiP}}
\newcommand{\pFormulas}{\mathcal{L}}

\renewcommand{\true}{\top}
\renewcommand{\false}{\bot}

\newcommand{\relsym}[1]{\thinspace{#1}\thinspace}
\newcommand{\knows}[2]{#1\relsym{\mathsf{k}}#2}

\newcommand{\limp}{\rightarrow}
\newcommand{\lequiv}{\leftrightarrow}

\newcommand{\K}[1]{\mathsf{K}_{#1}}

\newcommand{\proves}[4]{#1\relsym{:_{#3}^{#4}}#2}

\newcommand{\aModalFrame}{\mathfrak{S}}

\newcommand{\LiPded}{\vdash_{\LiP}}

\usepackage{colortbl} 
\renewcommand{\preorder}[1]{<_{#1}}
\renewcommand{\proves}[4]{#1\relsym{::_{#3}^{#4}}#2}

\newcommand{\LiiP}{\mathrm{LiiP}}

\newcommand{\LiiPded}{\vdash_{\LiiP}}
\newcommand{\LiiPplus}{\LiiP^{+}}
\newcommand{\LiiPplusDed}{\vdash_{\LiiPplus}}

\newcommand{\LiiPdedBis}{\mathrel{{\dashv}{\vdash}_{\LiiP}}}

\newcommand{\pproves}[4]{#1\relsym{:_{#3}^{#4}}#2}

\begin{document}
\title{Logic of Non-Monotonic Interactive Proofs\thanks{Work funded with 
		Grant AFR~894328 from the National Research Fund Luxembourg  
			cofunded under the Marie-Curie Actions of the European Commission (FP7-COFUND), and
			finalised during an invited stay at the Institute of Mathematical Sciences, Chennai, India.}\\
		\Large{(Formal Theory of Temporary Knowledge Transfer)}
	}
  \author{Simon Kramer\\[\jot]
  	%University of Luxembourg\\
	\texttt{simon.kramer@a3.epfl.ch}}
%\author{Simon Kramer}
%\institute{University of Luxembourg\\[\jot]
%			\texttt{simon.kramer@a3.epfl.ch}}
\maketitle
	
\begin{abstract}
	We propose a monotonic logic of internalised \emph{non-monotonic} or \emph{instant} interactive proofs (LiiP) and
		reconstruct an existing monotonic logic of internalised monotonic or persistent interactive proofs (LiP) as 
			a minimal conservative extension of LiiP.
	Instant interactive proofs effect
		a \emph{fragile} epistemic impact in their intended communities of peer reviewers that
			consists in the \emph{impermanent} induction of the knowledge of their proof goal 
				by means of the knowledge of the proof with the interpreting reviewer:
	If my peer reviewer knew my proof 
	then she would \emph{at least then (in that instant)} know that its proof goal is true.
	Their impact is fragile and their induction of knowledge impermanent in the sense of 
		being the case possibly only at the instant of learning the proof.
	This accounts for the important possibility of internalising proofs of statements 
		 whose 
				truth value can vary, 
					which, as opposed to invariant statements, 
						cannot have persistent proofs.	
	So instant interactive proofs effect
		a \emph{temporary} transfer of certain propositional knowledge (knowable \emph{ephemeral} facts) via 
			the transmission of certain individual knowledge (knowable \emph{non-monotonic} proofs) in 
				distributed systems of multiple interacting agents.

	\medskip
	
  \end{abstract}
  
%Vint Cerf: the importance of the ephemeral, Turing100 @ 3:00

\section{Introduction}
The subject matter of this paper is modal logic of interactive proofs, \ie 
	a novel logic of \emph{non-monotonic} or \emph{instant} interactive proofs (LiiP) as well as 
	an existing logic of monotonic or persistent interactive proofs (LiP) \cite{LiP}.
(We abbreviate interactivity-related adjectives with lower-case letters.)
The goal here is 
	to define LiiP axiomatically and semantically as well as 
	to reconstruct LiP as a minimal conservative extension of LiiP.
So for distributed and multi-agent systems, whose 
	states and thus truth of statements about states can vary, 
proof non-monotonicity (as in LiiP) is 
	in a logical sense more primitive than proof monotonicity (as in LiP).
In contrast, 
	proof monotonicity is perhaps more intuitive than proof non-monotonicity 
		within formal physical theories validated by experiment and surely 
		within mathematical theories known to be consistent.

Rephrasing \cite[Section~1.1]{NonMonotonicReasoning} model-theoretically, 
	the proof modality of LiiP internalises a non-monotonic notion of proof in the sense that    
		it can happen that  
			a proposition $\phi$ can be proved with a (non-monotonic) proof $M$ to an agent $a$ 
				in some system state $s$, but
				not anymore in some subsequent state $s'$ in which 
					$a$ will have learnt additional or lost previously learnt data $M'$.
See Appendix~\ref{section:ApplicationExamples} for formal application examples.
Like in LiP \cite{LiP},
	we understand interactive \emph{proofs as sufficient evidence} to 
		intended \emph{resource-unbounded} (though unable to guess) 
			proof- and signature-checking agents (designated verifiers).

\emph{Instant} interactive proofs effect
	a \emph{fragile} epistemic impact in their intended communities $\community$ of peer reviewers that
		consists in the \emph{impermanent} induction of the (propositional) knowledge (not only belief) of their proof goal $\phi$
			by means of the (individual) knowledge of the proof (the sufficient evidence) $M$ with 
				the designated interpreting reviewer $a:$
If $a$ knew my proof $M$ of $\phi$ 
then she would \emph{at least then (in that instant)} know that the proof goal $\phi$ is true.
By individual knowledge we mean
		knowledge in the sense of the transitive use of the verb ``to know,''
			here to know a message, such as the plaintext of an encrypted message.
Notation: $\knows{a}{M}$ for ``agent $a$ knows message $M$'' 
	(\cf Definition~\ref{definition:LiiPLanguage}).
This is the classic concept of knowledge \emph{de re} (``of a thing'') 
	made explicit for messages,  
	meaning taking them apart (analysing) and 
	putting them together (synthesising).
Whereas by propositional knowledge we mean
	knowledge in the sense of the use of the verb ``to know'' with a clause,  
		here to know that a statement is true, 
			such as that the plaintext of an encrypted message is (individually) unknown to potential adversaries. 
Notation: $\K{a}(\phi)$ for ``agent $a$ knows that $\phi$ (is true)''
	(\cf Fact~\ref{fact:KC}).
This is the classic concept of knowledge \emph{de dicto} (``of a fact'').\footnote{
	In a first-order setting, 
		knowledge \emph{de re} and \emph{de dicto} can
			be related in Barcan-laws \cite{MMFAAMAS}.}
(We distinguish individual and propositional knowledge with respect to the \emph{``object''} of knowledge [the know\emph{n}],
		\ie with respect to a message and clause, respectively.
However, individual as well as propositional knowledge can both be individual with respect to the \emph{subject} of knowledge [the know\emph{er}], 
	\ie an [individual] agent.)
With respect to belief, propositional knowledge essentially differs in that
		it is necessarily true whereas 
		belief is possibly false, as 
			commonly known and accepted \cite{MultiAgents}.
The epistemic impact of our instant interactive proofs is fragile and their induction of knowledge impermanent in the sense of 
	being the case possibly only at the instant of learning the proof.
This accounts for the important possibility of internalising proofs of statements, whose 
	truth value can vary, such as statements about system states, which, 
		as opposed to invariant statements, cannot have persistent proofs.
Proofs must (not) prove true (false) statements!
Standard examples of statements of variable truth value are 
	contingent (\eg elementary) facts (expressed as atomic formulas) and 
	characteristic formulas of states \cite{ModalModelTheory}.

In contrast \cite{LiP},
	the epistemic impact of \emph{persistent} interactive proofs is \emph{durable}  
	in the sense of 
		being the case necessarily at the instant of learning the proof \emph{and henceforth}, where
			time can be present implicitly (such as here) or explicitly (in future work).
In other words,
	when a persistent proof can prove a certain statement, 
		the proof will always be able to \emph{robustly} do so, 
			independently of whether or not more messages (data) than just the proof are learnt.

In sum, our instant interactive proofs effect  
	a transfer of propositional knowledge (knowable ephemeral facts) via 
		the transmission of certain individual knowledge (knowable non-monotonic proofs) in multi-agent 
		distributed systems.
That is,
	\emph{L(i)iP is a formal theory of (temporary) knowledge transfer.}
The overarching motivation for L(i)iP is 
		to serve in an intuitionistic foundation of interactive computation. 
See \cite{LiP} for a programmatic and methodological motivation.

\subsection{Contribution}\label{section:contribution}
Our technical contribution in this paper is fourfold. 
	For LiiP,
		we provide
			an adequate axiomatisation of its 
				oracle-computational and knowledge-constructive Kripke-semantics, and
			a minimal conservative extension LiiP$^{+}$ with a single monotonicity axiom schema making 
					LiiP$^{+}$ isomorphic to LiP. 
	For LiP,
		we provide 
			a substantially simplified semantic interface and
			a slightly simplified axiomatisation, which is 
	 			a nice side-effect of obtaining LiiP$^{+}$.
		
	The Kripke-semantics for LiiP (like for LiP \cite{LiP}) is knowledge-constructive in the sense that
		(\cf Fact~\ref{fact:KC})
		our interactive proofs 
			induce the knowledge of their proof goal (say $\phi$)
				in their intended interpreting agents (say $a$) such that
					the induced knowledge ($\K{a}(\phi)$) is knowledge in the sense of 
						the standard modal logic of knowledge S5
							\cite{Epistemic_Logic,MultiAgents,EpistemicLogicFiveQuestions}.
	Note that our agents here are still resource-\emph{un}bounded with respect to 
		individual and propositional knowledge, though
			they are still unable to guess that knowledge.
	(Recall that S5-agents are resource-unbounded, \ie logically omniscient.)
	Thus we give an epistemic explication of proofs, \ie
		an explication of proofs in terms of the epistemic impact that 
				they effect in their intended interpreting agents 
					(\ie the knowledge of their proof goal).
	Technically,
		we endow the proof modality with a standard Kripke-semantics \cite{ModalLogicSemanticPerspective}, but whose
			accessibility relation $\access{M}{a}{\community}$ we 
				first define constructively in terms of 
					elementary set-theoretic constructions,\footnote{in loose analogy with 
							the set-theoretically constructive rather than 
							the purely axiomatic definition of 
								numbers \cite{TheNumberSystems} or 
								ordered pairs (\eg the now standard definition by Kuratowski, 
									and other well-known definitions \cite{NotesOnSetTheory})} 
							namely as $\pAccess{M}{a}{\community}$,
				and then match to an abstract semantic interface in standard form (which 
					abstractly stipulates the characteristic properties of the accessibility relation
						\cite{ModalProofTheory}).
	We will say that $\pAccess{M}{a}{\community}$ \emph{exemplifies} 
		(or \emph{realises}) $\access{M}{a}{\community}$.
	(A simple example of a constructive definition of a modal accessibility is
		the well-known definition of epistemic accessibility as 
			state indistinguishability defined in terms of  
				equality of state projections \cite{Epistemic_Logic}.)
	Recall, set-theoretically constructive is different from 
		intuitionistically constructive!
	The Kripke-semantics for LiiP is oracle-computational in the sense that
		(\cf Definition~\ref{definition:SemanticIngredients})
		the individual proof knowledge (say $M$) can be thought of as being provided by 
			an imaginary computation oracle, which thus acts as a 
				hypothetical provider and imaginary epistemic source of our interactive proofs.
	The semantic interface of LiP here is simplified in the sense that
		we are able 
			to eliminate all \emph{a posteriori} constraints from the
				semantic interface in \cite{LiP} and thus 
			to manage with only standard, \emph{a priori} constraints, \ie stipulations.
	%
	%The notion of equality of proofs for LiiP and LiP presented here is richer than 
	%	the one for LiP in \cite{LiP} in the sense that
	%		the one in \cite{LiP} resides in the poorer structure of an idempotent commutative monoid.
	%
	%The extension of LiiP with abstract G\"odel-numberings is minimalistic in the sense that 
	%	we are able to generate the extension by 
	%		enriching the proof-term language of LiiP with 
	%			a single additional proof-term constructor for 
	%				\emph{formula quotation,} and by
	%		stipulating two simple axioms on the atomic propositions for proof-term knowledge.
	%
	%In particular, no new (compound) formula constructors are needed.
	%	
	%The extension is modular in the sense that
	%	we are able to generate the (single) characteristic law of the quoted-proof modality automatically from 
	%		a corresponding axiom for quoted-proof-term knowledge.
	%
	%Thus the resulting axiomatisation is automatically adequate (sound and complete) for 
	%	the intended semantics.

\subsection{Roadmap}
In the next section,
	we introduce our Logic of instant interactive Proofs (LiiP) axiomatically by means of 
				a compact closure operator 
			that induces the Hilbert-style proof system that we seek and
			that allows the simple generation of applica\-tion-specific extensions of LiiP 
				(\cf Appendix~\ref{section:ApplicationExamples}).
We then prove some useful (further-used) deducible laws within the obtained system.
Next, we introduce the set-theoretically constructive semantics and the abstract semantic interface for LiiP, and
	prove the axiomatic adequacy of the proof system with respect to this interface.
In the construction of the semantics,
	we again make use of a closure operator, but 
		this time on sets of proof terms.
Finally in Section~\ref{section:LiP},
	we reconstruct LiP as a minimal conservative extension of LiiP.
%in Section~\ref{section:LiPs},
%	we integrate LiP and LiiP in L[i]iP; and 
%in Section~\ref{section:Popper},
%	we extend LiiP to LiiPP.
%
%Finally in Section~\ref{section:NewTermTheories}, 
%	we present our notion of 
%		proof co-pairing and equality as well as 
%		abstract G\"odel-numberings and
%		subjectively random terms for all our logics.

\section{Logic of instant interactive Proofs}
The Logic of instant interactive Proofs (LiiP) provides
	a modal \emph{formula language} over a generic message \emph{term language}.
The formula language offers
		the propositional constructors, 
		a relational symbol `$\knows{}{}$' for constructing atomic propositions about individual knowledge (\eg $\knows{a}{M}$), and
		a modal constructor `$\proves{}{}{}{}$' for propositions about proofs
			(\eg $\proves{M}{\phi}{a}{\community}$).
The message language offers
	term constructors for 
		message \emph{pairing} and (not necessarily, but possibly cryptographically implemented) \emph{signing}.
(Cryptographic signature creation and verification is polynomial-time computable \cite{DigitalSignatures}.
See \cite{LiP} for other cryptographic constructors such as encryption and hashing.)
In brief, LiiP is a minimal modular extension of classical propositional logic with an interactively generalised
	additional operator (the proof modality) and 
	proof-term language (only two constructors, \emph{agents as proof- and signature-checkers}).
Note, the language of LiiP is identical to the one of LiP \cite{LiP}
	modulo the proof-modality notation, which
		in LiP is `$\pproves{}{}{}{}$'.
\begin{definition}[The language of LiiP]\label{definition:LiiPLanguage}
	Let
	\begin{itemize}
		\item $\agents\neq\emptyset$ designate a non-empty finite set of 
			\emph{agent names} $a$, $b$, $c$, \etc 
		\item $\community\subseteq\agents$ denote 
			(finite and not necessarily disjoint) communities (sets) of agents $a\in\agents$ (referred to by their name)
		\item $\messages \ni M\ \bnfeq\ 
			a \bnfor B \bnfor 
			\sign{M}{a} \bnfor
			\pair{M}{M}$ 
			designate 
			our language of \emph{message terms} $M$ \emph{over} $\agents$ with 
				(transmittable) agent names $a\in\agents$, 
				application-specific data $B$ (left blank here),
				signed messages $\sign{M}{a}$, and
				message pairs $\pair{M}{M}$ 
			
			(Messages must be grammatically well-formed, which
				yields an induction principle.
				So agent names $a$ are logical term constants, 
				the meta-variable $B$ just signals the possibility of an extended term language $\messages$, 
				$\sign{\cdot}{a}$ with $a\in\agents$ is a unary functional symbol, and
				$\pair{\cdot}{\cdot}$ a binary functional symbol.)
		\item $\mathcal{P}$ designate a denumerable set of \emph{propositional variables} $P$ constrained such that
			for all $a\in\agents$ and $M\in\messages$, 
				$(\knows{a}{M})\in\mathcal{P}$ (for ``$a$ knows $M$'') is 
					a distinguished variable, \ie 
						an \emph{atomic proposition}, (for \emph{individual} knowledge)
						
			(So, for $a\in\agents$, $\knows{a}{\cdot}$ is a unary relational symbol.)
		\item $\pFormulas\ni\phi \bnfeq P \bnfor 
				\neg\phi \bnfor 
				\phi\land\phi \bnfor 
				\colorbox[gray]{0.75}{$\proves{M}{\phi}{a}{\community}$}$ 
				designate our language of \emph{logical formulas} $\phi$, where
				$\proves{M}{\phi}{a}{\community}$ reads 
					``$M$ is a $\community\cup\set{a}$-\emph{reviewable proof} of $\phi$'' 
						in %the sense 
						that 
							``$M$ can prove $\phi$ to $a$ (\eg a designated verifying judge) and this %fact  
								is commonly accepted in the (pointed) community 
									$\community\cup\set{a}$ (\eg for $\community$ being a jury).''
	\end{itemize}
\end{definition}
Then LiiP has the following axiom and deduction-rule schemas, 
with grey-shading indicating the difference to LiP. 
\begin{definition}[The axioms and deduction rules of LiiP]\label{definition:AxiomsRules}
Let
	\begin{itemize}
		\item $\Gamma_{0}$ designate an adequate set of axioms for classical propositional logic
		\item $\Gamma_{1} \defeq \Gamma_{0} \cup \{$
			\begin{itemize}
				\item $\knows{a}{a}$\quad(knowledge of one's own name string)
				\item $\knows{a}{M}\limp\knows{a}{\sign{M}{a}}$\quad(\emph{personal} [the same $a$] signature \emph{synthesis})
				\item $\knows{a}{\sign{M}{b}}\limp\knows{a}{\pair{M}{b}}$\quad(\emph{universal} [any $a$ and $b$] signature \emph{analysis})
				\item $(\knows{a}{M}\land\knows{a}{M'})\lequiv\knows{a}{\pair{M}{M'}}$\quad([un]pairing)
				\item \colorbox[gray]{0.75}{$(\proves{M}{(\phi\limp\phi')}{a}{\community})\limp
				((\proves{M}{\phi}{a}{\community})\limp\proves{M}{\phi'}{a}{\community})$}%
				\quad(Kripke's law, K)
				\item $(\proves{M}{\phi}{a}{\community})\limp(\knows{a}{M}\limp\phi)$%
				\quad(epistemic truthfulness)
				\item \colorbox[gray]{0.75}{$\bigwedge_{b\in\community\cup\set{a}}((\underbrace{\proves{\pair{M}{b}}{\phi}{a}{\community}}_{\text{can prove}})\limp\proves{\sign{M}{a}}{(\underbrace{\knows{a}{M}\land\proves{M}{\phi}{a}{\community}}_{\text{does prove}})}{b}{\community\cup\set{a}})$}\newline
			(nominal [in $b$] peer review)
				\item $(\proves{M}{\phi}{a}{\community\cup\community'})\limp
												\proves{M}{\phi}{a}{\community}$\quad(group decomposition) \}
			\end{itemize}
			designate a set of \emph{axiom schemas.} 
	\end{itemize}
	Then, $\colorbox[gray]{0.75}{$\LiiP\defeq\Clo{}{}(\emptyset)$}\defeq\bigcup_{n\in\mathbb{N}}\Clo{}{n}(\emptyset)$, where 
		for all $\Gamma\subseteq\pFormulas:$
		\begin{eqnarray*}
					\Clo{}{0}(\Gamma) &\defeq& \Gamma_{1}\cup\Gamma\\
					\Clo{}{n+1}(\Gamma) &\defeq& 
						\begin{array}[t]{@{}l@{}}
							\Clo{}{n}(\Gamma)\ \cup\\
							\setst{\phi'}{\set{\phi,\phi\limp\phi'}\subseteq\Clo{}{n}(\Gamma)}\cup
								\quad\text{(\emph{modus ponens}, MP)}\\
							\setst{\proves{M}{\phi}{a}{\community}}{\phi\in\Clo{}{n}(\Gamma)}\cup
								\quad\text{(necessitation, N)}\\
							\colorbox[gray]{0.75}{$\setst{(\proves{M}{\phi}{a}{\community})\lequiv\proves{M'}{\phi}{a}{\community}}{(\knows{a}{M}\lequiv\knows{a}{M'})\in\Clo{}{n}(\Gamma)}$}\\
							\quad\text{(epistemic bitonicity)}.\label{page:EpistemicAntitonicity}
						\end{array}
				\end{eqnarray*}
		We call $\LiiP$ a \emph{base theory,} and
		$\Clo{}{}(\Gamma)$ an \emph{LiiP-theory} for any $\Gamma\subseteq\pFormulas$.
\end{definition}
\noindent
Notice the 
	logical order of LiiP, which is,
		due to propositions about (proofs of) propositions, \emph{higher-order propositional}.
Further, observe that 
	we assume
		the existence of a dependable mechanism for signing messages, which
		we model with the above synthesis and analysis axioms.\label{page:UnforgeableSignatures}
In \emph{trusted} multi-agent systems, signatures are unforg\emph{ed,} and thus
	such a mechanism is trivially given by 
		the inclusion of the sender's name in the sent message, or by
		the sender's sensorial impression on the receiver when communication is immediate.
In \emph{dis}trusted multi-agent systems (\eg the open Internet), 
	a practically unforge\emph{able} signature mechanism can be implemented with  
		classical \emph{certificate-based} or, more directly, with 
		\emph{identity-based} public-key cryptography \cite{DigitalSignatures}.
We also assume the existence of a pairing mechanism modelling finite sets.
Such a mechanism is required by the important application of
	communication (not only cryptographic) protocols \cite[Chapter~3]{SecurityEngineering}, in which
		concatenation of high-level data packets is associative, commutative, and idempotent.	
The key to the validity of K is that 
	we understand interactive proofs as \emph{sufficient evidence} for
		intended resource-unbounded proof-checking agents (who are though still unable to guess), 
			see \cite[Section~3.2.2]{LiP} for more details.
Next, the significance of epistemic truthfulness to interactivity is that
	in truly distributed multi-agent systems, 
		not all proofs are known by all agents, \ie 
			agents are not omniscient with respect to messages. 
Otherwise, why communicate with each other?
So there being a proof does not imply knowledge of that proof. 
When an agent $a$ does not know the proof and 
the agent cannot generate the proof \emph{ex nihilo} herself by guessing it, 
only communication from a peer, who thus acts as an oracle, can entail the knowledge of the proof with $a$. 
That is, provability and truth are necessarily concomitant in the non-interactive setting, 
whereas in interactive settings they are not necessarily so \cite{LiP}.
In nominal peer review,
	``can prove'' suggests the proof \emph{potentiality} of $\pair{M}{b}:$ 
		``if $a$ \emph{were to} know, \eg receive, $\pair{M}{b}$'' 
			(and thus know her potential interlocutor $b$'s name).
Whereas given $\sign{M}{a}$ to $b$, \eg 
	in an acknowledgement from $a$, 
		``does prove'' suggests the proof \emph{actuality} of $M:$
			``$a$ \emph{does} know, \eg did receive, $\pair{M}{b}$'', otherwise
				$a$ could not have signed $M$.
See the proof of Corollary~\ref{corollary:ConcreteAccessibility}.5
	for a semantic justification of the \emph{raison d'\^{e}tre} of $b$ in $\pair{M}{b}$.
Then, the justification for the necessitation rule (schema) is that
	in interactive settings,
		validities, and thus \emph{a fortiori} tautologies 
			(in the strict sense of validities of the propositional fragment), 
				are in some sense trivialities \cite{LiP}.
To see why,
	recall that
		modal validities are true in \emph{all} pointed models 
			(\cf Definition~\ref{definition:TruthValidity}), and thus
			not worth being communicated from one point to another in a given model, \eg 
				by means of specific interactive proofs.
(Nothing is logically more embarrassing than  
	talking in tautologies.)
Therefore, 
	validities deserve \emph{arbitrary} proofs.
What is worth being communicated are
	truths weaker than validities, namely 
		local truths in the standard model-theoretic sense 
			(\cf Definition~\ref{definition:TruthValidity}), which
				may not hold universally.
Otherwise why communicate with each other?
Finally,
	observe that
		epistemic bitonicity is a rule of \emph{logical modularity} that 
			allows the modular generation of structural modal laws from 
				equivalence term laws  
					(\cf Theorem~\ref{theorem:SomeUsefulDeducibleStructuralLaws}).

The grey-shading in Definition~\ref{definition:AxiomsRules} indicates that 
	the axioms and rules of LiiP 
		differ from those of LiP in exactly \label{page:LiiPvsLiP}
			Kripke's law, 
			nominal peer review, and 
			epistemic bitonicity (\cf \cite{LiP} and Section~\ref{section:LiP}).
In LiP, 
	these three LiiP-laws correspond 
		  to 
			the \emph{generalised} Kripke-law 
				$(\pproves{M}{(\phi\limp\phi')}{a}{\community})\limp
					((\pproves{M'}{\phi}{a}{\community})\limp\pproves{\pair{M}{M'}}{\phi'}{a}{\community})$, 
			(plain) peer review
				$(\pproves{M}{\phi}{a}{\community})\limp\bigwedge_{b\in\community\cup\set{a}}(\pproves{\sign{M}{a}}{(\knows{a}{M}\land\pproves{M}{\phi}{a}{\community})}{b}{\community\cup\set{a}})$, and 
			epistemic \emph{anti}tonicity
				``from $\knows{a}{M}\limp\knows{a}{M'}$ 
					deduce $(\pproves{M'}{\phi}{a}{\community})\limp\pproves{M}{\phi}{a}{\community}$'', 
				respectively.
The addition of the axiom schema 
	$$\boxed{(\proves{M}{\phi}{a}{\community})\limp\proves{\pair{M}{M'}}{\phi}{a}{\community}}$$ to LiiP 
		will result in a logic LiiP$^{+}$ that is isomorphic to LiP (\cf Theorem~\ref{theorem:LiiPplus}).
So in some sense,
	the essential difference between 
		instant proofs (proofs for at least an instant) and 
		persistent proofs (proofs for eternity) is distilled in this single additional law.
Following Art\"emov in \cite{JustificationLogic},
	this law can be interpreted as 
		Lehrer and Paxson's indefeasibility condition for justified true belief 
			\cite{LiP}.
In sum, 
	while 
		both LiP-proofs and LiiP-proofs are indefeasible 
			in the instant when they are learnt 
				(they induce knowledge, not only belief), 
		LiiP-proofs (LiP-proofs) are possibly (necessarily) (in)defeasible in the future of
			the instant in which they are learnt.

Now note the following macro-definitions: 
	$\true \defeq \knows{a}{a}$, 
	$\false \defeq \neg \true$, 
	$\phi \lor \phi' \defeq \neg (\neg \phi \land \neg \phi')$,
	$\phi \limp \phi' \defeq \neg \phi \lor \phi'$, and 
	$\phi \lequiv \phi' \defeq (\phi \limp \phi') \land (\phi' \limp \phi)$. 
	%$\knows{\community}{M} \defeq \bigwedge_{a\in\community}\knows{a}{M}$ 
	%	(``everybody in $\community$ knows $M$''), 
	%$\ignores{\community}{M} \defeq \bigwedge_{a\in\community}\neg(\knows{a}{M})$
	%	(``everybody in $\community$ ignores $M$''), 
	%and, more interestingly, \label{page:K}
	%	\begin{equation}\label{equation:K}
	%		\colorbox[gray]{0.75}{$\K{a}(\phi)\defeq\proves{a}{\phi}{a}{\emptyset}$, 
	%			for ``$a$ knows that $\phi$ [is true],''}
	%	\end{equation}
			%those in Table~\ref{table:ImportantMacroDefinitions}, of which  
				%the grey-shaded ones are 
	%			which is notably not definable with the intended meaning in LiP.
	%
	%Variations on our notions of interactive proof can also be macro-defined, \eg  
	%	with respect to \emph{reviewer communities} (by 
	%		conjunction with respect to their members and based on a policy of 
	%			either one (dis)proof for \emph{all} members 
	%			or one (dis)proof for \emph{each} member) and
	%	with respect to \emph{exclusive communities} (with respect to members only).
In the sequel, ``:iff'' abbreviates ``by definition, if and only if''.
\begin{proposition}[Hilbert-style proof system]\label{proposition:Hilbert}
	Let 
		\begin{itemize}
			\item $\Phi\LiiPded\phi$ \text{:iff} if $\Phi\subseteq\LiiP$ then $\phi\in\LiiP$ 
			\item $\phi\LiiPdedBis\phi'$ \text{:iff} $\set{\phi}\LiiPded\phi'$ and $\set{\phi'}\LiiPded\phi$
			\item $\LiiPded\phi$ \text{:iff} $\emptyset\LiiPded\phi.$
		\end{itemize}
		In other words, ${\LiiPded}\subseteq\powerset{\pFormulas}\times\pFormulas$ is a \emph{system of closure conditions} in the sense of 
				\cite[Definition~3.7.4]{PracticalFoundationsOfMathematics}.
		For example:
			\begin{enumerate}
				\item for all axioms $\phi\in\Gamma_{1}$, $\LiiPded\phi$
				\item for \emph{modus ponens}, $\set{\phi,\phi\limp\phi'}\LiiPded\phi'$
				\item for necessitation, $\set{\phi}\LiiPded\proves{M}{\phi}{a}{\community}$
				\item for epistemic bitonicity,
					$\set{\knows{a}{M}\lequiv\knows{a}{M'}}\LiiPded
						(\proves{M}{\phi}{a}{\community})\lequiv\proves{M'}{\phi}{a}{\community}$.
			\end{enumerate}
		(In the space-saving, horizontal Hilbert-notation ``$\Phi\LiiPded\phi$'', 
			$\Phi$ is not a set of hypotheses but a set of premises, \cf  
				\emph{modus ponens,} necessitation, and epistemic bitonicity.)
	Then $\LiiPded$ can be viewed as being defined by 
		a $\Clo{}{}$-induced Hilbert-style proof system.
	In fact 
			${\Clo{}{}}:\powerset{\pFormulas}\rightarrow\powerset{\pFormulas}$ is a \emph{standard consequence operator,} \ie
				a \emph{substitution-invariant compact closure operator.}
%					\begin{enumerate}
%						\item $\Gamma\subseteq\Clo{}{}(\Gamma)$\quad(extensivity)
%						\item if $\Gamma\subseteq\Gamma'$ then $\Clo{}{}(\Gamma)\subseteq\Clo{}{}(\Gamma')$\quad(monotonicity)
%						\item $\Clo{}{}(\Clo{}{}(\Gamma))\subseteq\Clo{}{}(\Gamma)$\quad(idempotency)
%						\item $\Clo{}{}(\Gamma)=\bigcup_{\Gamma'\in\powersetFinite{\Gamma}}\Clo{}{}(\Gamma')$\quad(compactness)
%						\item $\sigma[\Clo{}{}(\Gamma)]\subseteq\Clo{}{}(\sigma[\Gamma])$\quad(substitution invariance), 
%					\end{enumerate}
%					where $\sigma$ designates an arbitrary propositional $\pFormulas$-substitution.
\end{proposition}
\begin{proof}
	Like in \cite{LiP}.
	That a Hilbert-style proof system can be viewed as induced by 
		 a compact closure operator is well-known (\eg see \cite{WhatIsALogicalSystem});
	that $\Clo{}{}$ is indeed such an operator can be verified by 
		inspection of the inductive definition of $\Clo{}{}$; and
	substitution invariance follows from our definitional use of axiom \emph{schemas}.\footnote{%
		Alternatively to axiom schemas,
		we could have used 
			axioms together with an
			additional substitution-rule set
				$\setst{\sigma[\phi]}{\phi\in\Clo{}{n}(\Gamma)}$
		in the definiens of $\Clo{}{n+1}(\Gamma)$.}
\end{proof}
%Appendix~\ref{appendix:DeducibleLaws} contains some laws deducible within our Hilbert-style proof system.
%\section{Some deducible laws}\label{appendix:DeducibleLaws}

\begin{corollary}[Normality]\label{corollary:Normality}
		LiiP is a normal modal logic.
\end{corollary}
\begin{proof}
	Jointly by  
		Kripke's law, 
		\emph{modus ponens},
		necessitation (these by definition), and
		substitution invariance (\cf Proposition~\ref{proposition:Hilbert}).
\end{proof}

We are now going to present some useful (further-used), 
	deducible \emph{structural} laws of LiiP.
Here, 
	``structural'' means 
	``deducible exclusively from term axioms''.
The laws are enumerated in a (total) order that respects (but cannot reflect) their respective proof prerequisites.
The laws are also deducible in LiP, in the same order \cite{LiP}.
(All LiiP-deducible laws are also LiP-deducible, but not vice versa.)
\begin{theorem}[Some useful deducible structural laws]\label{theorem:SomeUsefulDeducibleStructuralLaws}\ 
	\begin{enumerate}
		\item $\LiiPded\knows{a}{\pair{M}{M'}}\limp\knows{a}{M}$\quad(left projection, 
				1-way $\Kcomb$-combinator property)
		\item $\LiiPded\knows{a}{\pair{M}{M'}}\limp\knows{a}{M'}$\quad(right projection)
		\item $\LiiPded\knows{a}{\pair{M}{M}}\lequiv\knows{a}{M}$\quad(pairing idempotency)
		\item $\LiiPded\knows{a}{\pair{M}{M'}}\lequiv\knows{a}{\pair{M'}{M}}$\quad(pairing commutativity)
		\item $\LiiPded(\knows{a}{M}\limp\knows{a}{M'})\lequiv(\knows{a}{\pair{M}{M'}}\lequiv\knows{a}{M})$\quad(neutral pair elements)
		\item $\LiiPded\knows{a}{\pair{M}{a}}\lequiv\knows{a}{M}$\quad(self-neutral pair element)
		\item $\LiiPded\knows{a}{\pair{M}{\pair{M'}{M''}}}\lequiv\knows{a}{\pair{\pair{M}{M'}}{M''}}$\quad(pairing associativity)
		\item $\LiiPded(\proves{\pair{M}{M}}{\phi}{a}{\community})\lequiv\proves{M}{\phi}{a}{\community}$\quad(proof idempotency)
		\item $\LiiPded(\proves{\pair{M}{M'}}{\phi}{a}{\community})\lequiv\proves{\pair{M'}{M}}{\phi}{a}{\community}$\quad(proof commutativity)
		\item $\set{\knows{a}{M}\limp\knows{a}{M'}}\LiiPded
					(\proves{\pair{M}{M'}}{\phi}{a}{\community})\lequiv\proves{M}{\phi}{a}{\community}$\quad(neutral proof elements)
		\item $\LiiPded(\proves{\pair{M}{a}}{\phi}{a}{\community})\lequiv\proves{M}{\phi}{a}{\community}$\quad(self-neutral proof element)
		\item $\LiiPded(\proves{\pair{M}{\pair{M'}{M''}}}{\phi}{a}{\community})\lequiv\proves{\pair{\pair{M}{M'}}{M''}}{\phi}{a}{\community}$\quad(proof associativity)
		\item $\LiiPded(\proves{\sign{M}{a}}{\phi}{a}{\community})\lequiv\proves{M}{\phi}{a}{\community}$\quad(self-signing idempotency)
	\end{enumerate}
\end{theorem}
\begin{proof} 
	Laws~1--7 and 13 are proved like in LiP \cite{LiP}, as   
		LiiP and LiP have identical term axioms.
	Law~8, 9, 11, and 12 follows immediately from Law~3, 4, 6, and 7, respectively by 
		epistemic bitonicity.
	For Law~10, 
		suppose that 
			$\LiiPded\knows{a}{M}\limp\knows{a}{M'}$.
		Hence $\LiiPded\knows{a}{\pair{M}{M'}}\lequiv\knows{a}{M}$ by 
			the law of neutral pair elements and
			propositional logic.
		Hence $\LiiPded\proves{\pair{M}{M'}}{\phi}{a}{\community}\lequiv\proves{M}{\phi}{a}{\community}$ by
			epistemic bitonicity.
\end{proof}
Like in LiP \cite{LiP}, 
	the preceding 1-way $\Kcomb$-combinator property and 
	the following simple corollary of Theorem~\ref{theorem:SomeUsefulDeducibleStructuralLaws} jointly establish the important fact that
		our communicating agents can be viewed as combinators in the sense of Combinatory Logic 
			viewed in turn as a (non-equational) theory of (message or proof) term reduction \cite{LambdaCalculusAndCombinators}.
(The converse of the above $\Kcomb$-combinator property does not hold.)
\begin{corollary}[$\Scomb$-combinator property]\ 
	\begin{enumerate}
		\item $\LiiPded\knows{a}{\pair{\pair{M}{M'}}{M''}}\lequiv\knows{a}{\pair{M}{\pair{M''}{\pair{M'}{M''}}}}$
		\item $\LiiPded(\proves{\pair{\pair{M}{M'}}{M''}}{\phi}{a}{\community})\lequiv
						\proves{\pair{M}{\pair{M''}{\pair{M'}{M''}}}}{\phi}{a}{\community}$
	\end{enumerate}
\end{corollary}
\begin{proof}
	1 follows jointly from idempotency (copy $M'''$), commutativity, and associativity of pairing; and
	2 follows jointly from 1 and epistemic bitonicity.
\end{proof}

We are going to present also some useful (further-used) 
	deducible \emph{logical} laws of LiiP.
Here, 
	``logical'' means 
	``not structural'' in the previously defined sense.
Also these laws 
	are enumerated in an order that respects their respective proof prerequisites, and
	are deducible in LiP in the same order \cite{LiP}.
\begin{theorem}[Some useful deducible logical laws]\label{theorem:SomeUsefulDeducibleLogicalLaws}\ 
		\begin{enumerate}
		\item $\set{\phi\limp\phi'}\LiiPded(\proves{M}{\phi}{a}{\community})\limp\proves{M}{\phi'}{a}{\community}$\quad(regularity)		
		\item $\set{\knows{a}{M}\lequiv\knows{a}{M'},\phi\limp\phi'}\LiiPded(\proves{M}{\phi}{a}{\community})\limp\proves{M'}{\phi'}{a}{\community}$ (biepistemic regul.)
		\item $\LiiPded((\proves{M}{\phi}{a}{\community})\land
						\proves{M}{\phi'}{a}{\community})\lequiv\proves{M}{(\phi\land\phi')}{a}{\community}$\quad(proof conjunctions \emph{bis})
		\item $\LiiPded((\proves{M}{\phi}{a}{\community})\lor\proves{M}{\phi'}{a}{\community})\limp
						\proves{M}{(\phi\vee\phi')}{a}{\community}$\quad(proof disjunctions \emph{bis})
		\item $\LiiPded\proves{M}{\true}{a}{\community}$\quad(anything can prove tautological truth)
		%\item $\LiiPded(\proves{a}{\phi}{a}{\community})\limp\phi$\quad(self-truthfulness)
		%\item $\phi\LiiPdedBis\proves{a}{\phi}{a}{\community}$\quad(self-truthfulness \emph{bis})
		%\item $\LiiPded\knows{a}{M}\limp\neg(\proves{M}{\false}{a}{\community})$\quad(nothing known can prove falsehood)
		%\item $\LiiPded\knows{a}{M}\limp((\proves{M}{\phi}{a}{\community})\limp\proofdiamond{M}{\phi}{a}{\community})$\quad(epistemic proof consistency)
		\item $\LiiPded\proves{\sign{M}{b}}{\knows{b}{M}}{a}{\community\cup\set{b}}$\quad(authentic knowledge)
		\item $\LiiPded\proves{M}{\knows{a}{M}}{a}{\emptyset}$\quad(self-knowledge)
		\item $\LiiPded(\proves{M}{\phi}{a}{\community\cup\community'})\limp
				((\proves{M}{\phi}{a}{\community})\land\proves{M}{\phi}{a}{\community'})$\quad(group decomposition \emph{bis})
		\item $\LiiPded(\proves{M}{\phi}{a}{\community\cup\set{a}})\lequiv(\proves{M}{\phi}{a}{\community})$
					\quad(self-neutral group element).
		\item $\LiPded\proves{M}{((\proves{M}{\phi}{a}{\community})\limp\phi)}{a}{\community}$\quad
				(self-proof of truthfulness)
		\item $\LiPded\proves{M}{(\neg(\proves{M}{\false}{a}{\community}))}{a}{\community}$\quad
				(self-proof of proof consistency)
		\item $\LiPded(\proves{M}{(\proves{M}{\phi}{a}{\community})}{a}{\community})\lequiv
				\proves{M}{\phi}{a}{\community}$\quad(modal idempotency)
\end{enumerate}
\end{theorem}
\begin{proof}
	Like in LiP \cite{LiP}.
\end{proof}
Like in LiP,
	the key to the validity of modal idempotency is that 
		each agent (\eg $a$) can act herself as proof-checker, 
			see \cite[Section~3.2.2]{LiP} for more details.

We now continue to 
	(re)present the constructive semantics for LiiP (\cf \cite[Section~2.2]{LiP}) and
	establish some important new and further-used 
		results about it.
The essential differences to the semantics of LiP are grey-shaded.
\begin{definition}[Semantic ingredients]\label{definition:SemanticIngredients}
For the knowledge-constructive model-theoretic study of LiiP let 
\begin{itemize}
	\item $\states$ designate the \emph{state space}---a set of \emph{system states} $s$
		\item $\msgs{a}:\states\rightarrow\powerset{\messages}$ designate 
			a \emph{raw-data extractor} that 
				extracts (without analysing) the (finite) set of messages from a system state $s$ that
					agent $a\in\agents$ has 
						either generated (assuming that only $a$ can generate $a$'s signature) or else 
						received \emph{as such} (not only as a strict subterm of another message)\label{page:RawData}; that is, $\msgs{a}(s)$ is $a$'s \emph{data base} in $s$
		\item $\clo{a}{s}:\powerset{\messages}\rightarrow\powerset{\messages}$ designate a \emph{data-mining operator} such that \label{page:DataMining}
			$\clo{a}{s}(\data)\defeq\clo{a}{}(\msgs{a}(s)\cup\data)\defeq\bigcup_{n\in\mathbb{N}}\clo{a}{n}(\msgs{a}(s)\cup\data)$, where for all $\data\subseteq\messages$:
				\begin{eqnarray*}
					\clo{a}{0}(\data) &\defeq& \set{a}\cup\data\\
					\clo{a}{n+1}(\data) &\defeq& 
						\begin{array}[t]{@{}l@{}}
							\clo{a}{n}(\data)\ \cup\\
							\setst{\pair{M}{M'}}{\set{M,M'}\subseteq\clo{a}{n}(\data)}\cup\quad\text{(pairing)}\\
							\setst{M, M'}{\pair{M}{M'}\in\clo{a}{n}(\data)}\cup\quad\text{(unpairing)}\\
							\setst{\sign{M}{a}}{M\in\clo{a}{n}(\data)}\cup\quad\text{(\emph{personal} signature \emph{synthesis})}\\
							\setst{\pair{M}{b}}{\sign{M}{b}\in\clo{a}{n}(\data)}\quad\text{(\emph{universal} signature \emph{analysis})}
						\end{array}
				\end{eqnarray*}
				%(For application-specific terms such as encryption, 
				%	we would have to add here the closure conditions corresponding to their characteristic term axioms.)
	\item \colorbox[gray]{0.75}{\parbox[t]{0.91\textwidth}{${\preorder{a}^{M}}\subseteq\states\times\states$ designate a %(non-reflexive, though not irreflexive) 
	\emph{data preorder} on states such that
		for all $s,s'\in\states$,
			$s\preorder{a}^{M}s'$ :iff $\clo{a}{s}(\set{M})=\clo{a}{s'}(\emptyset)$,
				were $M$ can be viewed as \emph{oracle input} in addition to  
					$a$'s \emph{individual-knowledge base} $\clo{a}{s}(\emptyset)$ (\cf also \cite[Section~2.2]{LiP})}}
			%(The reader is invited to consider the effects of encryption on closure here.)
	\item \colorbox[gray]{0.75}{\parbox[t]{0.91\textwidth}{${\preorder{\community}^{M}}\defeq(\bigcup_{a\in\community}{\preorder{a}^{M}})^{++}$, where 
			`$^{++}$' designates the closure operation of so-called \emph{generalised transitivity}
				in the sense that 
					 ${\preorder{\community}^{M}}\circ{\preorder{\community}^{M'}}\subseteq{\preorder{\community}^{\pair{M}{M'}}}$}} 
				%as suggested in Proposition~\ref{proposition:GT} for `$\preorder{a}^{M}$'}}
	\item ${\indist{a}{}{}}\defeq{\preorder{a}^{a}}$ designate an equivalence relation of 
		\emph{state indistinguishability}
	\item ${\pAccess{M}{a}{\community}}\subseteq\states\times\states$ designate 
		a \emph{\textbf{concretely constructed} accessibility relation}---short, \emph{\textbf{concrete} accessibility}---for 
			the proof modality such that for all $s,s'\in\states$,  
			\begin{eqnarray*}
				s\pAccess{M}{a}{\community}s' &\defiff& %$\preorder{\community\cup\set{a}}^{M}$ 
								s'\in\hspace{-6ex} 
					\bigcup_{\scriptsize 
					\begin{array}{@{}c@{}}
						\text{\colorbox[gray]{0.75}{$s\preorder{\community\cup\set{a}}^{M}\tilde{s}$} and
						}\\[0.5\jot] 
						M\in\clo{a}{\tilde{s}}(\emptyset) 
					\end{array}
					}\hspace{-5.5ex}
					[\tilde{s}]_{\indist{a}{}{}}\\
				&\text{(iff}& 
					\text{there is $\tilde{s}\in\states$ \st 
							$s\preorder{\community\cup\set{a}}^{M}\tilde{s}$ and
							$M\in\clo{a}{\tilde{s}}(\emptyset)$ and
							$\indist{a}{\tilde{s}}{s'}$).}
			\end{eqnarray*}
\end{itemize}
\end{definition}
Note that
	the data-mining operator
		$\clo{a}{}:\powerset{\messages}\rightarrow\powerset{\messages}$ is a compact closure operator, which
		induces a \emph{data-derivation relation} ${\derives{a}{}{}}\subseteq\powerset{\messages}\times\messages$ such that $\derives{a}{M}{\data}$ :iff $M\in\clo{a}{}(\data)$, which  
	(1) has the compactness and (2) the cut property,  
	(3) is decidable in deterministic polynomial time in the size of $\data$ and $M$, and
	(4) induces a Scott information system of information tokens $M$ \cite{LiP}.
Fact~\ref{fact:KC} establishes the knowledge-constructiveness of 
	our Kripke-model for LiiP (\cf Definition~\ref{definition:KripkeModel}).
\begin{fact}[Kripke-model knowledge-constructiveness]\label{fact:KC}
	\begin{multline*}
		\text{for all $s'\in\states$, if $s\pAccess{M}{a}{\community}s'$
							then $(\aModalFrame, \mathcal{V}), s'\models\phi$ if and only if}\\
			\text{$\begin{array}[t]{@{}l@{}}
					\text{for all $\check{s}\in\states$, 
						if $s\preorder{\community\cup\set{a}}^{M}\check{s}$ 
					then $(\aModalFrame, \mathcal{V}), \check{s}\models\knows{a}{\hspace{-2.5ex}\underbrace{M}_{\text{\begin{tabular}{@{}c@{}}
						sufficient\\[-1\jot] 
						evidence
						\end{tabular}}}}\hspace{-2ex}\limp\K{a}(\hspace{-2.5ex}\underbrace{\phi}_{\hspace{0ex}\text{\begin{tabular}{@{}c@{}}
		induced\\[-1\jot] 
		knowledge
		\end{tabular}}}\hspace{-2.5ex})$}
				\end{array}$,}
					\end{multline*}
		where the standard epistemic modality $\K{a}$ is defined like in \cite{MultiAgents} as 
		 $$\begin{array}{@{}l@{}}
		 	(\aModalFrame, \mathcal{V}), \check{s}\models\K{a}(\phi)\quad\text{:iff}
		 	\quad\text{for all $s'\in\states$, if $\indist{a}{\check{s}}{s'}$ then 
				$(\aModalFrame, \mathcal{V}), s'\models\phi$.}
			\end{array}$$
\end{fact}
\begin{proof}
	By elementary-logical transformations of the definiens of $\pAccess{M}{a}{\community}$.
\end{proof}

\begin{lemma}\label{lemma:HKKHK}
	If $s\preorder{a}^{M}s'$ then $s'\preorder{a}^{M}s'$.
\end{lemma}
\begin{proof}
	Consider that 
		when $s\preorder{a}^{M}s'$,
			$M\in\clo{a}{s'}(\emptyset)$, and thus
			 	$\clo{a}{s'}(\set{M})=\clo{a}{s'}(\emptyset)$.
\end{proof}

\begin{proposition}[Restricted reflexivity]\label{proposition:RestrictedReflexivity}\  
	\begin{enumerate}
		\item $s\preorder{a}^{a}s$\quad(self-reflexivity)
		\item biconditional reflexivity:
			\begin{enumerate}
				\item $s\preorder{a}^{M}s$ if and only if $M\in\clo{a}{s}(\emptyset)$
				\item $s\preorder{a}^{M}s$ if and only if there is $s'\in\states$ such that $s'\preorder{a}^{M}s$
			\end{enumerate}
		%\item seriality: just add $M$ to $s$ and obtain $s'$
	\end{enumerate}
\end{proposition}
\begin{proof}
	For 1, 
		consider that 
			$a\in\clo{a}{s}(\emptyset)$, and thus
				$\clo{a}{s}(\set{a})=\clo{a}{s}(\emptyset)$.
	For 2.a, 
		inspect the proof of Lemma~\ref{lemma:HKKHK}.
	For the forward-direction of 2.b, 
		take $s$ as $s'$; and
			for the backward-direction apply Lemma~\ref{lemma:HKKHK}.
\end{proof}

\begin{proposition}[Self-symmetry]\label{proposition:SelfSymmetry}
	$$\text{If $s\preorder{a}^{a}s'$ then $s'\preorder{a}^{a}s$.}$$
\end{proposition}
\begin{proof}
	By expansion of the definition of `$\preorder{a}^{a}$' and the symmetry of equality.
\end{proof}

\begin{proposition}[Generalised transitivity]\label{proposition:GT} 
	$$\text{If $s\preorder{a}^{M}s'$ and $s'\preorder{a}^{M'}s''$ 
				then $s\preorder{a}^{\pair{M}{M'}}s''$.}$$
\end{proposition}
\begin{proof}
	Let $s,s'\in\states$ and suppose that 
			$s\preorder{a}^{M}s'$ and 
			$s'\preorder{a}^{M'}s''$.
	Thus:
		\begin{enumerate}
			\item $\clo{a}{s}(\set{M})=\clo{a}{s'}(\emptyset)$; 
				thus $M\in\clo{a}{s'}(\emptyset)$, thus:
				\begin{enumerate}
					\item $M\in\clo{a}{s'}(\set{M'})$ by closure monotonicity ($\emptyset\subseteq\set{M'}$), 
					\item $\clo{a}{s'}(\emptyset)=\clo{a}{s'}(\set{M})$, 
							thus $\clo{a}{s}(\set{M})=\clo{a}{s'}(\set{M})$, and 
							hence\newline $\clo{a}{s}(\set{\pair{M}{M'}})=\clo{a}{s'}(\set{\pair{M}{M'}})$;
					\end{enumerate}
			\item $\clo{a}{s'}(\set{M'})=\clo{a}{s''}(\emptyset)$;  
			thus $M'\in\clo{a}{s''}(\emptyset)$, 
				thus $\clo{a}{s''}(\emptyset)=\clo{a}{s''}(\set{M'})$, 
				thus $\clo{a}{s'}(\set{M'})=\clo{a}{s''}(\set{M'})$, and
				hence $\clo{a}{s'}(\set{\pair{M}{M'}})=\clo{a}{s''}(\set{\pair{M}{M'}})$.
		\end{enumerate}
	Hence: 
		\begin{itemize}
			\item $M\in\clo{a}{s''}(\emptyset)$ by 1.a and the first assertion in 2 , 
				thus $\pair{M}{M'}\in\clo{a}{s''}(\emptyset)$ by the second assertion in 2 and pairing closure, 
				thus $\clo{a}{s''}(\emptyset)=\clo{a}{s''}(\set{\pair{M}{M'}})$; 
			\item $\clo{a}{s}(\set{\pair{M}{M'}})=\clo{a}{s''}(\set{\pair{M}{M'}})$ by 1.b and 2.
		\end{itemize}
	Hence $\clo{a}{s}(\set{\pair{M}{M'}})=\clo{a}{s''}(\emptyset)$, and
	thus $s\preorder{a}^{\pair{M}{M'}}s''$ by definition.
\end{proof}

\begin{corollary}[Transitivity]\label{corollary:Transitivity}
	$$\text{If $s\preorder{a}^{M}s'$ and $s'\preorder{a}^{M}s''$ 
				then $s\preorder{a}^{M}s''$.}$$

\end{corollary}
\begin{proof} 
	Directly from Proposition~\ref{proposition:GT}  
		by the fact that $\clo{a}{s}(\set{\pair{M}{M}})=\clo{a}{s}(\set{M})$.
\end{proof}
So as announced in Definition~\ref{definition:SemanticIngredients}, 
	`$\preorder{a}^{M}$' is indeed a (non-reflexive) pre-order, and 
	`$\preorder{a}^{a}$' indeed an equivalence relation 
		(\cf Proposition~\ref{proposition:RestrictedReflexivity}.i and \ref{proposition:SelfSymmetry}).

\begin{definition}[Message ordering and equivalence]\ 
	\begin{itemize}
		\item $M\sqsubseteq_{a}^{s}M'$ :iff 
			if $M\in\clo{a}{s}(\emptyset)$ then $M'\in\clo{a}{s}(\emptyset)$
		\item $M\equiv_{a}^{s}M'$ :iff $M\sqsubseteq_{a}^{s}M'$ and $M'\sqsubseteq_{a}^{s}M$
		\item $M\sqsubseteq_{a}M'$ :iff for all $s\in\states$, $M\sqsubseteq_{a}^{s}M'$
		\item $M\equiv_{a}M'$ :iff for all $s\in\states$, $M\equiv_{a}^{s}M'$
	\end{itemize}
\end{definition}

\begin{fact} 
	${\sqsubseteq_{a}^{s}}\subseteq\messages\times\messages$ is a pre- but not a partial order.
\end{fact}

\begin{proposition}[Conditional stability]\label{proposition:ConditionalStability}
		$$\text{If $M\equiv_{a}M'$ then ${\preorder{a}^{M}}={\preorder{a}^{M'}}$.}$$
\end{proposition}
\begin{proof}
	Suppose that 
		for all $s''\in\states$, $M\in\clo{a}{s''}(\emptyset)$ if and only if  $M'\in\clo{a}{s''}(\emptyset)$, and
		let $s,s'\in\states$.
	For the $\subseteq$-part, suppose that $s\preorder{a}^{M}s'$, \ie 
		$\clo{a}{s}(\set{M})=\clo{a}{s'}(\emptyset)$, and thus 
		$M\in\clo{a}{s'}(\emptyset)$.
	Hence:
		\begin{enumerate}
			\item $M'\in\clo{a}{s'}(\emptyset)$ by particularisation of the first hypothesis, and
				$\pair{M}{M'}\in\clo{a}{s'}(\emptyset)$ by pairing closure; and 
				thus $\clo{a}{s'}(\set{\pair{M}{M'}})=\clo{a}{s'}(\emptyset);$
			\item $M\in\clo{a}{s}(\emptyset)$ if and only if $M'\in\clo{a}{s}(\emptyset)$ by particularisation of the first hypothesis, thus $M\in\clo{a}{s}(\set{M'})$ if and only if $M'\in\clo{a}{s}(\set{M'})$, thus $M\in\clo{a}{s}(\set{M'})$, and
				thus $\clo{a}{s}(\set{M'})=\clo{a}{s}(\set{\pair{M}{M'}});$
			\item $\clo{a}{s'}(\set{M})=\clo{a}{s'}(\emptyset)$, thus
					$\clo{a}{s'}(\set{M})=\clo{a}{s}(\set{M})$, and thus
					$\clo{a}{s'}(\set{\pair{M}{M'}})=\clo{a}{s}(\pair{M}{M'})$.
		\end{enumerate}
	Hence $\clo{a}{s}(\set{M'})=\clo{a}{s'}(\emptyset)$ by 1, 2, and 3.
	And symmetrically for the $\supseteq$-part.
\end{proof}

\begin{proposition}[Communal lifting]\label{proposition:CommuncalLifting}\ 
	\begin{enumerate}
		%\item seriality: there is $s'\in\states$ such that $s\preorder{\community\cup\set{a}}^{M}s'$
		\item If $\community\subseteq\community'$ 
				then ${\preorder{\community}^{M}}\subseteq{\preorder{\community'}^{M}}$\quad(communal monotonicity).
		\item If $M\in\clo{a}{s}(\emptyset)$ then $s\preorder{\community\cup\set{a}}^{M}s$\quad(conditional reflexivity).
		\item If $M\equiv_{a}M'$ then ${\preorder{\community\cup\set{a}}^{M}}={\preorder{\community\cup\set{a}}^{M'}}$\quad(conditional stability).
	\end{enumerate}
\end{proposition}
\begin{proof}
	1 follows directly from definitions, 
	2 from 1 and Proposition~\ref{proposition:RestrictedReflexivity}.ii.a, and
	3 from Proposition~\ref{proposition:ConditionalStability} and 
		the definition of `$\preorder{\community\cup\set{a}}^{M}$' and `$\preorder{\community\cup\set{a}}^{M'}$'.
\end{proof}

\begin{proposition}[Signature property]\label{proposition:SignatureProperty}
	$$\text{If $s\preorder{\community}^{\sign{M}{a}}s'$ then $M\in\clo{a}{s'}(\emptyset)$.}$$
\end{proposition}
\begin{proof}
	Let $s,s'\in\states$ and 
	suppose that $s\preorder{\community}^{\sign{M}{a}}s'$.
	Thus there is $b\in\community$ such that $s\preorder{b}^{\sign{M}{a}}s'$.
	Hence $\sign{M}{a}\in\clo{b}{s'}(\emptyset)$ by biconditional reflexivity (\cf Proposition~\ref{proposition:RestrictedReflexivity}.ii.a).
	But then also $M\in\clo{a}{s'}(\emptyset)$ by 
		the unforgeability of signatures 
			(\cf the closure conditions of 
					personal/universal signature synthesis/analysis).
	That is, 
		nobody else than $a$ can have generated $\sign{M}{a}$, and 
			thus $a$ also knows $M$.
	(Otherwise suppose that somebody else has, and derive a contradiction.)
\end{proof}

%\begin{proposition}[Euclideanness]\label{proposition:Euclideanness}
%	$$\text{If $s\preorder{a}^{M}s'$ and $s\preorder{a}^{M}s''$ 
%	then $s'\preorder{a}^{M}s''$.}$$
%\end{proposition}
%\begin{proof}
%	Let 
%		$s,s',s''\in\states$ and 
%		suppose that 
%			$s\preorder{a}^{M}s'$ (\ie $\clo{a}{s}(\set{M})=\clo{a}{s'}(\emptyset)$, 
%				and thus $\clo{a}{s'}(\emptyset)=\clo{a}{s'}(\set{M})$) and
%			$s\preorder{a}^{M}s''$ (\ie $\clo{a}{s}(\set{M})=\clo{a}{s''}(\emptyset)$).
%	%
%	Hence $\clo{a}{s''}(\emptyset)=\clo{a}{s'}(\emptyset)$ 
%		(and thus $\clo{a}{s'}(\set{M})=\clo{a}{s''}(\emptyset)$).
%	%
%	Hence $s'\preorder{a}^{M}s''$.
%\end{proof}
%
%\begin{proposition}[Weak Euclideanness]\label{proposition:WeakEuclideanness}
%	$$\text{If $\indist{a}{s}{s'}$ and $s\preorder{a}^{M}s''$ 
%	then $s'\preorder{a}^{M}s''$.}$$
%\end{proposition}
%\begin{proof}
%	Let 
%		$s,s',s''\in\states$ and 
%		suppose that 
%			$\indist{a}{s}{s'}$ (\ie $\clo{a}{s}(\set{a})=\clo{a}{s'}(\emptyset)$, 
%				thus $\clo{a}{s}(\emptyset)=\clo{a}{s'}(\emptyset)$, and
%				thus $\clo{a}{s}(\set{M})=\clo{a}{s'}(\set{M})$) and
%			$s\preorder{a}^{M}s''$ (\ie $\clo{a}{s}(\set{M})=\clo{a}{s''}(\emptyset)$).
%	%
%	Hence $\clo{a}{s'}(\set{M})=\clo{a}{s''}(\emptyset)$.
%	%
%	Thus $s'\preorder{a}^{M}s''$.
%\end{proof}

\begin{corollary}[Concrete accessibility]\label{corollary:ConcreteAccessibility}\ 
\begin{enumerate}
	%\item seriality: 
	\item If $\community\subseteq\community'$ 
			then ${\pAccess{M}{a}{\community}}\subseteq{\pAccess{M}{a}{\community'}}$\quad(communal monotonicity).
	\item If $M\equiv_{a}M'$ then ${\pAccess{M}{a}{\community}}={\pAccess{M'}{a}{\community}}$\quad(conditional stability).
	\item If $M\in\clo{a}{s}(\emptyset)$
			then $s\pAccess{M}{a}{\community}s$\quad(conditional reflexivity).
	\item If $s\pAccess{\sign{M}{b}}{a}{\community}s'$ then $M\in\clo{b}{s'}(\emptyset)$\quad(signature property).
	\item For all $b\in\community\cup\set{a}$,  
		$({\pAccess{\sign{M}{a}}{b}{\community\cup\set{a}}}\circ{\pAccess{M}{a}{\community}})\subseteq{\pAccess{\pair{M}{b}}{a}{\community}}$\quad(communal transitivity).
%	\item If $s\pAccess{M}{a}{\emptyset}s'$ and $s\pAccess{M}{a}{\emptyset}s''$
%			then $s'\pAccess{M}{a}{\emptyset}s''$\quad(singleton $\set{a}$ Euclideanness).
%	\item If $s\pAccess{a}{a}{\emptyset}s'$ and $s\pAccess{M}{a}{\emptyset}s''$
%			then $s'\pAccess{M}{a}{\emptyset}s''$\quad(weak singleton Euclideanness).			
\end{enumerate}
\end{corollary}
\begin{proof}
	1--4 follow by inspection of 
		definitions and  
		Proposition~\ref{proposition:CommuncalLifting} and 
		\ref{proposition:SignatureProperty}.
	For 5, 
		suppose that 
			$b\in\community\cup\set{a}$ and
		let $s,s',s''\in\states$.
	Further 
		suppose that 
			$s\pAccess{\sign{M}{a}}{b}{\community\cup\set{a}}s'$ and 
			$s'\pAccess{M}{a}{\community}s''$.
	That is,
		(there is $\tilde{s}\in\states$ such that  
			$s\preorder{\community\cup\set{a}\cup\set{b}}^{\sign{M}{a}}\tilde{s}$ and
			$\sign{M}{a}\in\clo{b}{\tilde{s}}(\emptyset)$ and
			$\indist{b}{\tilde{s}}{s'}$) and
		(there is $\tilde{s}'\in\states$ such that  
			$s'\preorder{\community\cup\set{a}}^{M}\tilde{s}'$ and
			$M\in\clo{a}{\tilde{s}'}(\emptyset)$ and
			$\indist{a}{\tilde{s}'}{s''}$).
	Hence, 
		$s\preorder{\community\cup\set{a}}^{\sign{M}{a}}\tilde{s}$ by 
			the first supposition and communal monotonicity ($\community\cup\set{a}\cup\set{b}=\community\cup\set{a}$), and also 
		$\tilde{s}\preorder{b}^{b}s'$ by definition (\cf second supposition).
	Hence consecutively, 
		$\tilde{s}\preorder{\community\cup\set{a}}^{b}s'$ by the first supposition and communal monotonicity 
			($\set{b}\subseteq\community\cup\set{a}$), 
		$s\preorder{\community\cup\set{a}}^{\pair{\sign{M}{a}}{b}}s'$ by 
			generalised transitivity,
		$s\preorder{\community\cup\set{a}}^{\pair{\pair{\sign{M}{a}}{b}}{M}}\tilde{s}'$ by the third supposition and again generalised transitivity, 
		$s\preorder{\community\cup\set{a}}^{\pair{M}{b}}\tilde{s}'$ by conditional stability
			($\pair{\pair{\sign{M}{a}}{b}}{M}\equiv_{a}\pair{M}{b}$), and thus finally 
		$s\pAccess{\pair{M}{b}}{a}{\community}s''$ by again the third supposition.
\end{proof}

\begin{definition}[Kripke-model]\label{definition:KripkeModel}
We define the \emph{satisfaction relation} `$\models$' for 
		LiiP in Table~\ref{table:SatisfactionRelation}, %(for LiiPP, see Section~\ref{section:Popper}), 
	\begin{table}[t]
	\centering
	\caption{Satisfaction relation}
	\smallskip
	\fbox{$\begin{array}{@{}rcl@{}}
		(\aModalFrame, \mathcal{V}), s\models P &\text{:iff}& s\in\mathcal{V}(P)\\[\jot]
		(\aModalFrame, \mathcal{V}), s\models\neg\phi &\text{:iff}& \text{not $(\aModalFrame, \mathcal{V}), s\models\phi$}\\[\jot]
		(\aModalFrame, \mathcal{V}), s\models\phi\land\phi' &\text{:iff}& \text{$(\aModalFrame, \mathcal{V}), s\models\phi$ and $(\aModalFrame, \mathcal{V}), s\models\phi'$}\\[\jot]
		(\aModalFrame, \mathcal{V}), s\models\proves{M}{\phi}{a}{\community} &\text{:iff}& 
			\begin{array}[t]{@{}l@{}}
				\text{for all $s'\in\states$, }
			 	\text{if $s\access{M}{a}{\community}s'$ then $(\aModalFrame, \mathcal{V}), s'\models\phi$}
			\end{array}
	\end{array}$}
	\label{table:SatisfactionRelation}
	\end{table}
where 
	\begin{itemize}
		\item $\mathcal{V}:\mathcal{P}\rightarrow\powerset{\states}$ designates a usual \emph{valuation function,} yet
			partially predefined such that for all $a\in\agents$ and $M\in\messages$,
				$$\mathcal{V}(\knows{a}{M})\defeq\setst{s\in\states}{M\in\clo{a}{s}(\emptyset)}$$
				
				(If agents are Turing-machines 
					then $a$ knowing $M$ can be understood as $a$ being able to parse $M$ on its tape.)
		\item $\aModalFrame\defeq(\states,\set{\access{M}{a}{\community}}_{M\in\messages,a\in\agents,\community\subseteq\agents})$
			designates a (modal) \emph{frame} for LiiP with 
		an \emph{\textbf{abstractly constrained} accessibility relation}---short, \emph{\textbf{abstract}  accessibility}---${\access{M}{a}{\community}}\subseteq\states\times\states$ for 
			the proof modality such that---the \emph{semantic interface:}\label{page:AbstractProofAccessibility} 
			\begin{itemize}
	\item if $\community\subseteq\community'$ 
			then ${\access{M}{a}{\community}}\subseteq{\access{M}{a}{\community'}}$
	\item \colorbox[gray]{0.75}{if $M\equiv_{a}M'$ then ${\access{M}{a}{\community}}={\access{M'}{a}{\community}}$}
	\item if $M\in\clo{a}{s}(\emptyset)$
			then $s\access{M}{a}{\community}s$
	\item \colorbox[gray]{0.75}{if $s\access{\sign{M}{b}}{a}{\community}s'$ then $M\in\clo{b}{s'}(\emptyset)$}
	\item \colorbox[gray]{0.75}{for all $b\in\community\cup\set{a}$,  
		$({\access{\sign{M}{a}}{b}{\community\cup\set{a}}}\circ{\access{M}{a}{\community}})\subseteq{\access{\pair{M}{b}}{a}{\community}}$}
			\end{itemize}
	\item $(\aModalFrame,\mathcal{V})$ designates a (modal) \emph{model} for LiiP.
	\end{itemize}
\end{definition}
Looking back, 
	we recognise that Corollary~\ref{corollary:ConcreteAccessibility} actually establishes the important fact that
		our concrete accessibility $\pAccess{M}{a}{\community}$ in 
		Definition~\ref{definition:SemanticIngredients} realises 
			all the properties stipulated by  
		our abstract accessibility $\access{M}{a}{\community}$ in Definition~\ref{definition:KripkeModel};
		we say that 
		$$\text{$\pAccess{M}{a}{\community}$ \emph{exemplifies} (or \emph{realises}) $\access{M}{a}{\community}$.}$$
Further, observe that 
	LiiP (like LiP) has a Herbrand-style semantics, \ie
			logical constants (agent names) and 
			functional symbols (pairing, signing) are self-interpreted rather than 
		interpreted in terms of (other, semantic) constants and functions.
This simplifying design choice spares our framework from 
	the additional complexity that would arise from term-variable assignments \cite{FOModalLogic}, which in turn
		keeps our models propositionally modal.
Our choice is admissible because our individuals (messages) are finite.
(Infinitely long ``messages'' are non-messages; they can never be completely received, \eg
	transmitting irrational numbers as such is impossible.)

\begin{theorem}[Axiomatic adequacy]\label{theorem:Adequacy}\ 
	$\LiiPded$ is \emph{adequate} for $\models$, \ie:
	\begin{enumerate}
		\item if $\LiiPded\phi$ then $\models\phi$\quad(axiomatic soundness)
		\item if $\models\phi$ then $\LiiPded\phi$\quad(semantic completeness).
	\end{enumerate}
\end{theorem}
\begin{proof}
	Both parts can be proved with standard means:
	soundness follows as usual from 
		the admissibility of the axioms and rules 
			(\cf Appendix~\ref{appendix:AxiomaticSoundness}); and 
			%Proposition~\ref{proposition:AxiomAndRuleAdmissibility}
	completeness follows by means of the classical construction of canonical models,
		using Lindenbaum's construction of maximally consistent sets 
			(\cf Appendix~\ref{appendix:LiiPCompleteness}).
\end{proof}

\section{LiP as an extension of LiiP}\label{section:LiP}
In this section,
	we reconstruct LiP 
	syntactically, as a minimal conservative extension of LiiP with
		one simplified and 
		one additional axiom schema, as well as 
	semantically, with a simplified semantic interface that   
		has none of the \emph{a posteriori} constraints from \cite{LiP} but  
		only standard, \emph{a priori} constraints, \ie stipulations.

\begin{theorem}\label{theorem:LiiPplus}
	Define the $\LiiP$-theory $$\LiiPplus\defeq\Clo{}{}(\set{\underbrace{(\proves{M}{\phi}{a}{\community})\limp
			\proves{\pair{M}{M'}}{\phi}{a}{\community}}_{\text{proof extension}}}),$$
				where $\Clo{}{}$ is as in Definition~\ref{definition:AxiomsRules}.
	Then $\LiiPplus$ is isomorphic to $\LiP$, in symbols,  
		$$\LiiPplus\cong\LiP.$$
	In particular,
		the generalised Kripke law GK as mentioned before and below is deducible in $\LiiPplus$,  
			and thus we need only stipulate the simpler standard Kripke law K for LiP, like for LiiP.
	Moreover,
		\emph{alternatively to adding} the axiom schema of proof extension to LiiP,
			we could \emph{equivalently replace} the primitive rule schema of epistemic bitonicity in LiiP with  
				the stronger one of epistemic antitonicity.
\end{theorem}
\begin{proof} The isomorphism consists in simply switching between 
	proof-modality notations, which
		in $\LiP$ is `$\pproves{}{}{}{}$' and 
		in $\LiiPplus$ `$\proves{}{}{}{}$'.
	Then, as already mentioned on Page~\pageref{page:LiiPvsLiP},
		LiiP and LiP differ in the following corresponding axiom and deduction-rule schemas: 
			Kripke's law K versus the generalised Kripke-law GK,
			nominal peer review (NPR) versus plain peer review, and 
			epistemic bitonicity versus epistemic antitonicity---see below.
	Note that in the sequel   
		PL abbreviates ``(Classical) Propositional Logic,'' and
		$\LiiPplusDed$ is defined similarly to $\LiiPded$.
	\begin{itemize}
		\item GK (\cf Line 7) becomes deducible:
			\begin{enumerate}
				\item$\LiiPplusDed(\proves{M}{(\phi\limp\phi')}{a}{\community})\limp
							\proves{\pair{M}{M'}}{(\phi\limp\phi')}{a}{\community}$\hfill \textbf{\emph{proof extension}}
				\item$\LiiPplusDed(\proves{\pair{M}{M'}}{(\phi\limp\phi')}{a}{\community})\limp
							((\proves{\pair{M}{M'}}{\phi}{a}{\community})\limp\proves{\pair{M}{M'}}{\phi'}{a}{\community})$\hfill K 
				\item$\LiiPplusDed(\proves{M}{(\phi\limp\phi')}{a}{\community})\limp
							((\proves{\pair{M}{M'}}{\phi}{a}{\community})\limp\proves{\pair{M}{M'}}{\phi'}{a}{\community})$\hfill 1, 2 PL
				\item$\LiiPplusDed(\proves{M'}{\phi}{a}{\community})\limp
							\proves{\pair{M'}{M}}{\phi}{a}{\community}$\hfill \textbf{\emph{proof extension}}
				\item$\LiiPplusDed(\proves{\pair{M'}{M}}{\phi}{a}{\community})\lequiv
							\proves{\pair{M}{M'}}{\phi}{a}{\community}$\hfill proof commutativity
				\item$\LiiPplusDed(\proves{M'}{\phi}{a}{\community})\limp
							\proves{\pair{M}{M'}}{\phi}{a}{\community}$\hfill 4, 5, PL
				\item$\LiiPplusDed(\proves{M}{(\phi\limp\phi')}{a}{\community})\limp
							((\proves{M'}{\phi}{a}{\community})\limp\proves{\pair{M}{M'}}{\phi'}{a}{\community})$\hfill 3, 6, PL.
			\end{enumerate}
		\item plain peer review (\cf Line 3) becomes deducible:
			\begin{enumerate}
				\item$\LiiPplusDed\bigwedge_{b\in\community\cup\set{a}}((\proves{M}{\phi}{a}{\community})\limp
							\proves{\pair{M}{b}}{\phi}{a}{\community})$\hfill \textbf{\emph{proof extension}}
				\item$\LiiPplusDed\bigwedge_{b\in\community\cup\set{a}}((\proves{\pair{M}{b}}{\phi}{a}{\community})\limp\proves{\sign{M}{a}}{(\knows{a}{M}\land\proves{M}{\phi}{a}{\community})}{b}{\community\cup\set{a}})$\hfill NPR 
				\item$\LiiPplusDed(\proves{M}{\phi}{a}{\community})\limp\bigwedge_{b\in\community\cup\set{a}}(\proves{\sign{M}{a}}{(\knows{a}{M}\land\proves{M}{\phi}{a}{\community})}{b}{\community\cup\set{a}})$\hfill 1, 2, PL.
			\end{enumerate}
		\item epistemic antitonicity (\cf Line 8) becomes deducible:
			\begin{enumerate}
				\item\quad$\LiiPplusDed\knows{a}{M}\limp\knows{a}{M'}$\hfill hyp.
				\item\quad$\LiiPplusDed(\proves{\pair{M}{M'}}{\phi}{a}{\community})\lequiv
							\proves{M}{\phi}{a}{\community}$\hfill 1, neutral proof elememts
				\item\quad$\LiiPplusDed(\proves{M'}{\phi}{a}{\community})\limp
							\proves{\pair{M'}{M}}{\phi}{a}{\community}$\hfill \textbf{\emph{proof extension}}
				\item\quad$\LiiPplusDed(\proves{\pair{M'}{M}}{\phi}{a}{\community})\lequiv
							\proves{\pair{M}{M'}}{\phi}{a}{\community}$\hfill proof commutativity
				\item\quad$\LiiPplusDed(\proves{M'}{\phi}{a}{\community})\limp
							\proves{\pair{M}{M'}}{\phi}{a}{\community}$\hfill 3, 4, PL
				\item\quad$\LiiPplusDed(\proves{M'}{\phi}{a}{\community})\limp
							\proves{M}{\phi}{a}{\community}$\hfill\hfill 2, 5, PL
				\item if $\LiiPplusDed\knows{a}{M}\limp\knows{a}{M'}$ 
							then $\LiiPplusDed(\proves{M'}{\phi}{a}{\community})\limp
									\proves{M}{\phi}{a}{\community}$\hfill 1--6, PL
				\item$\set{\knows{a}{M}\limp\knows{a}{M'}}\LiiPplusDed
							(\proves{M'}{\phi}{a}{\community})\limp
								\proves{M}{\phi}{a}{\community}$\hfill 7, def.
			\end{enumerate}
	\end{itemize}
	Conversely, that is, assuming epistemic antitonicity, 
		proof extension is directly deducible from jointly this assumption and (pair) left projection, 
			like in LiP \cite{LiP}.
\end{proof}

\begin{corollary}[Simplified semantic interface for LiP]
A simplified semantic interface for LiP is given by 
	the one for LiiP in Definition~\ref{definition:KripkeModel} but 
		with the abstract accessibility ${\access{M}{a}{\community}}\subseteq\states\times\states$ 
			being constrained   
		\begin{itemize}
		\item such that
			\fbox{if $M\sqsubseteq_{a}M'$ then ${\access{M}{a}{\community}}\subseteq{\access{M'}{a}{\community}}$\quad(proof monotonicity)}
			
		\textbf{\emph{instead of}} being constrained by conditional stability;  
		\item \emph{or alternatively} such that
		\fbox{${\access{\pair{M}{M'}}{a}{\community}}\subseteq{\access{M}{a}{\community}}$\quad(pair splitting)} 
		
		\textbf{\emph{in addition to}} being constrained by conditional stability.
		\end{itemize}
\end{corollary}
\begin{proof}
	It is straightforward to check that 
		the semantic constraints of proof monotonicity and pair splitting correspond to 
			the syntactic laws of epistemic antitonicity and proof extension, respectively,
				which are interdeducible (\cf Theorem~\ref{theorem:LiiPplus}).
\end{proof}

\section{Conclusion}
We have proposed LiiP with 
	as main contributions those described in Section~\ref{section:contribution}.
The notion of non-monotonic proofs captured by LiiP has the advantage of being
	not only operational thanks to our proof-theoretic definition 
	but also declarative thanks to our complementary model-theoretic definition, which
		gives a constructive epistemic semantics to these proofs in the sense of
			explicating \emph{what} (knowledge) they effect in agents in the instant of their reception, 
				complementing thereby 
					the (operational) axiomatics, which explicates \emph{how} they do so. 

We conclude by mentioning \cite{BaltagRenneSmets} as a piece of related work.
There, 
	the authors present a \emph{resource-bounded implicit-single-agent} but dynamic logic of 
		defeasible (and thus non-monotonic) evidence-based S4-knowledge, where 
		they use a particular primitive $Et$ for 
			the implicit-agent's knowledge of evidence terms $t$.
The authors' atomic proposition $Et$ is a particular and strongly resource-bounded analog of 
	my atomic proposition $\knows{a}{M}$ for 
		an arbitrary agent $a$'s knowledge of message terms $M$. 
$Et$ is strongly resource-bounded in the sense that
	the term axioms for $Et$ are axioms for term decomposition but not for term composition.
Similar restrictions could be made for $\knows{a}{M}$, but
	we opine that they would be too strong.
At least \emph{some} amount of term composition capabilities should be conceded also to resource-bounded agents.
The authors' use of $Et$ is crucial for their contribution, who 
	know but must have accidentally not acknowledged the contribution of $\knows{a}{M}$ to $Et$.
See \cite{LiP} for historical references of my uses of $\knows{a}{M}$ in 
	logics of explicit evidence/justification/proof.
The addition of atomic propositions $\knows{a}{M}$ 
	to languages of explicit evidence/justification/proof 
		will probably play a similarly important role as 
the addition of atomic propositions $x\in S$ 
	to the language of first-order logic (resulting in Set Theory).

\bibliographystyle{alpha}
%\bibliography{/Users/simonkramer/Documents/Work/Sources/BibTeX/bibliography}

\appendix
\section{Axiomatic-adequacy proof}
	\newcommand{\canrel}[3]{\mathrel{_{#1}\negthinspace\mathrm{C}_{#2}^{#3}}}
	\newcommand{\canVal}{\mathcal{V}_{\mathsf{C}}}

\subsection{Axiomatic soundness}\label{appendix:AxiomaticSoundness}	
\begin{definition}[Truth \& Validity \cite{ModalLogicSemanticPerspective}]\label{definition:TruthValidity}\  
	\begin{itemize}
	\item The formula $\phi\in\pFormulas$ is \emph{true} (or \emph{satisfied}) 
		in the model $(\aModalFrame,\mathcal{V})$ at the state $s\in\states$ 
			:iff $(\aModalFrame,\mathcal{V}), s\models\phi$.
	\item The formula $\phi$ is \emph{satisfiable} in the model $(\aModalFrame,\mathcal{V})$ 
			:iff there is $s\in\states$ such that 
				$(\aModalFrame,\mathcal{V}), s\models\phi$.
	\item The formula $\phi$ is \emph{globally true} (or \emph{globally satisfied}) 
		in the model $(\aModalFrame,\mathcal{V})$, 
			written $(\aModalFrame,\mathcal{V})\models\phi$, :iff 
				for all $s\in\states$, $(\aModalFrame,\mathcal{V}),s\models\phi$.
	\item The formula $\phi$ is \emph{satisfiable}  
			:iff there is a model $(\aModalFrame,\mathcal{V})$ and a state $s\in\states$ such that 
				$(\aModalFrame,\mathcal{V}),s\models\phi$.	
	\item The formula $\phi$ is \emph{valid}, written $\models\phi$, :iff 
			for all models $(\aModalFrame,\mathcal{V})$, $(\aModalFrame,\mathcal{V})\models\phi$.
	\end{itemize}
\end{definition}

\begin{proposition}[Admissibility of LiiP-specific axioms and rules]\label{proposition:AxiomAndRuleAdmissibility}\ 
	\begin{enumerate}
		\item $\models\knows{a}{a}$
		\item $\models\knows{a}{M}\limp\knows{a}{\sign{M}{a}}$
		\item $\models\knows{a}{\sign{M}{b}}\limp\knows{a}{\pair{M}{b}}$
		\item $\models(\knows{a}{M}\land\knows{a}{M'})\lequiv\knows{a}{\pair{M}{M'}}$
		\item $\models(\proves{M}{(\phi\limp\phi')}{a}{\community})\limp
				((\proves{M}{\phi}{a}{\community})\limp\proves{M}{\phi'}{a}{\community})$
		\item $\models(\proves{M}{\phi}{a}{\community})\limp(\knows{a}{M}\limp\phi)$
		\item $\models\bigwedge_{b\in\community\cup\set{a}}((\proves{\pair{M}{b}}{\phi}{a}{\community})\limp\proves{\sign{M}{a}}{(\knows{a}{M}\land\proves{M}{\phi}{a}{\community})}{b}{\community\cup\set{a}})$
		\item $\models(\proves{M}{\phi}{a}{\community\cup\community'})\limp\proves{M}{\phi}{a}{\community}$
		\item If $\models\phi$ then $\models\proves{M}{\phi}{a}{\community}$
		\item If $\models\knows{a}{M}\lequiv\knows{a}{M'}$ then 
					$\models(\proves{M}{\phi}{a}{\community})\lequiv\proves{M'}{\phi}{a}{\community}$.
	\end{enumerate}
\end{proposition}
\begin{proof}
	1--4 are immediate;
	5 and 9 hold by the fact that LiiP has a standard Kripke-semantics; 
	6 follows directly from the conditional reflexivity of `$\access{M}{a}{\community}$', 	
	8 directly from the communal monotonicity of `$\access{M}{a}{\community}$', and
	10 directly from the conditional stability of `$\access{M}{a}{\community}$'.
	Finally,
		7 follows jointly from 
			the signature and 
			%the conditional-stability, and 
			the communal-transitivity property of `$\access{M}{a}{\community}$'---as follows:
		let $(\aModalFrame,\mathcal{V})$ designate an arbitrary LiiP-model and 
		let $s\in\states$.
		First,  
			let $b\in\community\cup\set{a}$ and
			suppose that $(\aModalFrame,\mathcal{V}),s\models\proves{\pair{M}{b}}{\phi}{a}{\community}$.
		Second, 
			let $s'\in\states$ and 
			suppose that $s\access{\sign{M}{a}}{b}{\community\cup\set{a}}s'$.
		Hence $M\in\clo{a}{s'}(\emptyset)$ by the signature property, and
			thus $(\aModalFrame,\mathcal{V}),s'\models\knows{a}{M}$ by definition.
		Third,  
			let $s''\in\states$ and 
			suppose that $s'\access{M}{a}{\community}s''$.
		Hence, 
			%$s'\access{\sign{M}{a}}{a}{\community}s''$ by conditional stability 
			%	($\sign{M}{a}\equiv_{a}M$),
			%$s\access{\pair{\sign{M}{a}}{b}}{a}{\community}s''$ by 
			%	the first and second supposition and communal transitivity, 
			%$s\access{\pair{M}{b}}{a}{\community}s''$ again by conditional stability 
			%	($\pair{\sign{M}{a}}{b}\equiv_{a}\pair{M}{b}$), and
			$s\access{\pair{M}{b}}{a}{\community}s''$ by 
				the first, second, and third supposition and communal transitivity.
		Hence $(\aModalFrame,\mathcal{V}),s''\models\phi$ by the first supposition.
		Thus $(\aModalFrame,\mathcal{V}),s'\models\proves{M}{\phi}{a}{\community}$ by 
				discharge of the third supposition.
		Hence $(\aModalFrame,\mathcal{V}),s'\models\knows{a}{M}\land\proves{M}{\phi}{a}{\community}$.
		Finally, consecutively discharging the remaining three suppositions,
			$(\aModalFrame,\mathcal{V}),s\models\proves{\sign{M}{a}}{(\knows{a}{M}\land\proves{M}{\phi}{a}{\community})}{b}{\community\cup\set{a}}$, then 
			$(\aModalFrame,\mathcal{V}),s\models(\proves{\pair{M}{b}}{\phi}{a}{\community})\limp\proves{\sign{M}{a}}{(\knows{a}{M}\land\proves{M}{\phi}{a}{\community})}{b}{\community\cup\set{a}}$, and then 
			$(\aModalFrame,\mathcal{V}),s\models\bigwedge_{b\in\community\cup\set{a}}((\proves{\pair{M}{b}}{\phi}{a}{\community})\limp\proves{\sign{M}{a}}{(\knows{a}{M}\land\proves{M}{\phi}{a}{\community})}{b}{\community\cup\set{a}})$.

\end{proof}

\subsection{Semantic completeness}\label{appendix:LiiPCompleteness}	
%\subsection{Completeness of $\LiiPded$}\label{appendix:LiiPCompleteness}
	For all $\phi\in\pFormulas$, 
		if $\models\phi$ then $\LiiPded\phi$.
\begin{proof}
	Let
		\begin{itemize}
			\item $\mathcal{W}$ designate the set of all maximally LiiP-consistent sets\footnote{*
				A set $W$ of LiiP-formulas is maximally LiiP-consistent :iff 
					$W$ is LiiP-consistent and 
					$W$ has no proper superset that is LiiP-consistent.
				A set $W$ of LiiP-formulas is LiiP-consistent :iff 
					$W$ is not LiiP-inconsistent.
				A set $W$ of LiiP-formulas is LiiP-inconsistent :iff 
					there is a finite $W'\subseteq W$ such that $((\bigwedge W')\limp\false)\in\LiiP$.
				Any LiiP-consistent set can be extended to a maximally LiiP-consistent set by means of  
					the Lindenbaum Construction \cite[Page~90]{ModalProofTheory}.
				A set is maximally LiiP-consistent if and only if 
				the set of logical-equivalence classes of the set is an ultrafilter of
				the Lindenbaum-Tarski algebra of LiiP \cite[Page~351]{AlgebrasAndCoalgebras}.
				The canonical frame is isomorphic to the ultrafilter frame of that Lindenbaum-Tarski algebra 
					\cite[Page~352]{AlgebrasAndCoalgebras}.}
			\item for all $w,w'\in\mathcal{W}$,
				$w\canrel{M}{a}{\community}w'$ :iff $\setst{\phi\in\pFormulas}{\proves{M}{\phi}{a}{\community}\in w}\subseteq w'$
			\item for all $w\in\mathcal{W}$, $w\in\canVal(P)$ :iff $P\in w$.
		\end{itemize}
	Then \newcommand{\canModel}{\mathfrak{M}_{\mathsf{C}}}
		$\canModel\defeq
			(\mathcal{W},\set{\canrel{M}{a}{\community}}_{M\in\messages,a\in\agents,\community\subseteq\agents},\canVal)$
		designates the \emph{canonical model} for LiiP.
	Following Fitting \cite[Section~2.2]{ModalProofTheory}, 
	the following useful property of $\canModel$,  
		$$\boxed{$\text{for all $\phi\in\pFormulas$ and $w\in\mathcal{W}$,
			$\phi\in w$ if and only if $\canModel,w\models\phi$,}$}$$
	the so-called \emph{Truth Lemma}, can be proved by induction on the structure of $\phi$:
	\begin{enumerate}
		\item Base case ($\phi\defeq P$ for $P\in\mathcal{P}$). 
			For all $w\in\mathcal{W}$,
				$P\in w$ if and only if $\canModel,w\models P$, 
					by definition of $\canVal$.
		\item Inductive step ($\phi\defeq \neg\phi'$ for $\phi'\in\pFormulas$).
			Suppose that
				for all $w\in\mathcal{W}$,
					$\phi'\in w$ if and only if $\canModel,w\models\phi'$.
			Further let
				$w\in\mathcal{W}$.
			Then, 
				$\neg\phi'\in w$ if and only if $\phi'\not\in w$ --- $w$ is consistent ---
				if and only if $\canModel,w\not\models\phi'$ --- by the induction hypothesis ---
				if and only if $\canModel,w\models\neg\phi'$.
		\item Inductive step ($\phi\defeq \phi'\land\phi''$ for $\phi',\phi''\in\pFormulas$).
			Suppose that 
				for all $w\in\mathcal{W}$,
					$\phi'\in w$ if and only if $\canModel,w\models\phi'$, and that
				for all $w\in\mathcal{W}$,
					$\phi''\in w$ if and only if $\canModel,w\models\phi''$.
			Further let
				$w\in\mathcal{W}$.
			Then, 
				$\phi'\land\phi''\in w$ if and only if 
					($\phi'\in w$ and $\phi''\in w$), because $w$ is maximal.
			Now suppose that
				$\phi'\in w$ and $\phi''\in w$.
			Hence, 
				$\canModel,w\models\phi'$ and 
				$\canModel,w\models\phi''$, by the induction hypotheses, and 
			thus $\canModel,w\models\phi'\land\phi''$.
			Conversely, suppose that
				$\canModel,w\models\phi'\land\phi''$.
			Then,
				$\canModel,w\models\phi'$ and 
				$\canModel,w\models\phi''$.
			Hence,
				$\phi'\in w$ and $\phi''\in w$, by the induction hypotheses.
			Thus,
				($\phi'\in w$ and $\phi''\in w$) if and only if
				($\canModel,w\models\phi'$ and 
				$\canModel,w\models\phi''$).
			Whence
				$\phi'\land\phi''\in w$ if and only if
				($\canModel,w\models\phi'$ and 
				$\canModel,w\models\phi''$), by transitivity.
		\item Inductive step ($\phi\defeq\proves{M}{\phi'}{a}{\community}$ for 
				$M\in\messages$, $a\in\agents$, $\community\subseteq\agents$, and $\phi'\in\pFormulas$).
			\begin{flushleft}
			\nn{4.1} for all $w\in\mathcal{W}$,
						$\phi'\in w$ if and only if $\canModel,w\models\phi'$\hfill ind.\ hyp.\\[\jot]
			\nn{4.2}\quad $w\in\mathcal{W}$\hfill hyp.\\[2\jot]
			\nn{4.3}\qquad $\proves{M}{\phi'}{a}{\community}\in w$\hfill hyp.\\[\jot]
			\nn{4.4}\qquad\quad	$w'\in\mathcal{W}$\hfill hyp.\\[\jot]
			\nn{4.5}\qquad\qquad $w\canrel{M}{a}{\community}w'$\hfill hyp.\\[\jot]
			\nn{4.6}\qquad\qquad $\setst{\phi''\in\pFormulas}{\proves{M}{\phi''}{a}{\community}\in w}\subseteq w'$\hfill 4.5\\[\jot]
			\nn{4.7}\qquad\qquad $\phi'\in\setst{\phi''\in\pFormulas}{\proves{M}{\phi''}{a}{\community}\in w}$\hfill 4.3, 4.6\\[\jot]
			\nn{4.8}\qquad\qquad $\phi'\in w'$\hfill 4.6, 4.7\\[\jot]
			\nn{4.9}\qquad\qquad $\canModel,w'\models\phi'$\hfill 4.1, 4.4, 4.8\\[\jot]
			\nn{4.10}\qquad\quad if $w\canrel{M}{a}{\community}w'$ then $\canModel,w'\models\phi'$\hfill 4.5--4.9\\[\jot]
			\nn{4.11}\qquad for all $w'\in\mathcal{W}$, 
					if $w\canrel{M}{a}{\community}w'$ 
					then $\canModel,w'\models\phi'$\hfill 4.4--4.10\\[\jot]
			\nn{4.12}\qquad $\canModel,w\models\proves{M}{\phi'}{a}{\community}$\hfill 4.11\\[2\jot]
			\nn{4.13}\qquad $\proves{M}{\phi'}{a}{\community}\not\in w$\hfill hyp.\\[\jot]
			\nn{4.14}\qquad\quad $\mathcal{F}=\setst{\phi''\in\pFormulas}{\proves{M}{\phi''}{a}{\community}\in w}\cup\set{\neg\phi'}$\hfill hyp.\\[\jot]
			\nn{4.15}\qquad\qquad $\mathcal{F}$ is LiiP-inconsistent\hfill hyp.\\[\jot]
			\nn{4.16}\qquad\qquad there is $\set{\proves{M}{\phi_{1}}{a}{\community},\ldots,\proves{M}{\phi_{n}}{a}{\community}}\subseteq w$ such that\\
			\nn{}\qquad\qquad $\LiiPded(\phi_{1}\land\ldots\land\phi_{n}\land\neg\phi')\limp\false$\hfill 4.14, 4.15\\[\jot]
			\nn{4.17}\qquad\qquad\quad $\set{\proves{M}{\phi_{1}}{a}{\community},\ldots,\proves{M}{\phi_{n}}{a}{\community}}\subseteq w$ and\\
			\nn{}\qquad\qquad\quad $\LiiPded(\phi_{1}\land\ldots\land\phi_{n}\land\neg\phi')\limp\false$\hfill hyp.\\[\jot]
			\nn{4.18}\qquad\qquad\quad $\LiiPded(\phi_{1}\land\ldots\land\phi_{n})\limp\phi'$\hfill 4.17\\[\jot]
			\nn{4.19}\qquad\qquad\quad $\LiiPded(\proves{M}{(\phi_{1}\land\ldots\land\phi_{n})}{a}{\community})\limp\proves{M}{\phi'}{a}{\community}$\hfill 4.18, regularity\\[\jot]
			\nn{4.20}\qquad\qquad\quad $\LiiPded((\proves{M}{\phi_{1}}{a}{\community})\land\ldots\land(\proves{M}{\phi_{n}}{a}{\community}))\limp\proves{M}{\phi'}{a}{\community}$\hfill 4.19\\[\jot]
			\nn{4.21}\qquad\qquad\quad $\proves{M}{\phi'}{a}{\community}\in w$\hfill 4.17, 4.20, $w$ is maximal\\[\jot]
			\nn{4.22}\qquad\qquad\quad false\hfill 4.13, 4.21\\[\jot]
			\nn{4.23}\qquad\qquad false\hfill 4.16, 4.17--4.22\\[\jot]
			\nn{4.24}\qquad\quad $\mathcal{F}$ is LiiP-consistent\hfill 4.15--4.23\\[\jot]
			\nn{4.25}\qquad\quad there is $w'\supseteq\mathcal{F}$ \st $w'$ is maximally LiiP-consistent\hfill 4.24\\[\jot]
			\nn{4.26}\qquad\qquad $\mathcal{F}\subseteq w'$ and $w'$ is maximally LiiP-consistent\hfill hyp.\\[\jot]
			\nn{4.27}\qquad\qquad $\setst{\phi''\in\pFormulas}{\proves{M}{\phi''}{a}{\community}\in w}\subseteq\mathcal{F}$\hfill 4.14\\[\jot]
			\nn{4.28}\qquad\qquad $\setst{\phi''\in\pFormulas}{\proves{M}{\phi''}{a}{\community}\in w}\subseteq w'$\hfill 4.26, 4.27\\[\jot]
			\nn{4.29}\qquad\qquad $w\canrel{M}{a}{\community}w'$\hfill 4.28\\[\jot]
			\nn{4.30}\qquad\qquad $w'\in\mathcal{W}$\hfill 4.26\\[\jot]
			\nn{4.31}\qquad\qquad $\neg\phi'\in\mathcal{F}$\hfill 4.14\\[\jot]
			\nn{4.32}\qquad\qquad $\neg\phi'\in w'$\hfill 4.26, 4.31\\[\jot]
			\nn{4.33}\qquad\qquad $\phi'\not\in w'$\hfill 4.26 ($w'$ is LiiP-consistent), 4.32\\[\jot]
			\nn{4.34}\qquad\qquad $\canModel,w'\not\models\phi'$\hfill 4.1, 4.33\\[\jot]
			\nn{4.35}\qquad\qquad there is $w'\in\mathcal{W}$ \st 
				$w\canrel{M}{a}{\community}w'$ and $\canModel,w'\not\models\phi'$\hfill 4.29, 4.34\\[\jot]
			\nn{4.36}\qquad\qquad $\canModel,w\not\models\proves{M}{\phi'}{a}{\community}$\hfill 4.35\\[\jot]
			\nn{4.37}\qquad\quad $\canModel,w\not\models\proves{M}{\phi'}{a}{\community}$\hfill 4.25, 4.26--4.36\\[\jot]
			\nn{4.38}\qquad $\canModel,w\not\models\proves{M}{\phi'}{a}{\community}$\hfill 4.14--4.37\\[2\jot]
			\nn{4.39}\quad $\proves{M}{\phi'}{a}{\community}\in w$ if and only if $\canModel,w\models\proves{M}{\phi'}{a}{\community}$\hfill 4.3--4.12, 4.13--4.38\\[\jot]
			\nn{4.40} for all $w\in\mathcal{W}$,
				$\proves{M}{\phi'}{a}{\community}\in w$ if and only if $\canModel,w\models\proves{M}{\phi'}{a}{\community}$\hfill 4.2--4.39\\[\jot]
		\end{flushleft}
%		\item Inductive step for $\phi\defeq\exists \vec{m}(\knows{a}{(M[\vec{m}])} \land \phi'[\vec{m}])$.
%			%
%			Let $\knows{a}{(M[\vec{M}/\vec{m}])}\in\mathcal{P}$, and
%			let $\phi'[\vec{M}/\vec{m}]$ designate a LiiP-formula.
%			%
%			Suppose that 
%				for all $w\in\mathcal{W}$,
%					$\phi'[\vec{M}/\vec{m}]\in w$ if and only if $\canModel,w\models\phi'[\vec{M}/\vec{m}]$.
%			Further let 
%				$w\in\mathcal{W}$.
%			%
%			\begin{flushleft}
%				$\exists \vec{m}(\knows{a}{(M[\vec{m}])} \land \phi'[\vec{m}])\in w$\\[\jot]
%				\quad iff there is $\vec{M}\in\vec{\messages}$ \st\\ 
%					\qquad$\knows{a}{(M[\vec{M}/\vec{m}])} \land \phi'[\vec{M}/\vec{m}]\in w$\hfill $w$ is maximal\\[\jot] 
%				\quad iff there is $\vec{M}\in\vec{\messages}$ \st\\ 
%					\qquad$\knows{a}{(M[\vec{M}/\vec{m}])}\in w$ and 
%						$\phi'[\vec{M}/\vec{m}]\in w$\hfill $w$ is maximal\\[\jot] 
%				\quad iff there is $\vec{M}\in\vec{\messages}$ \st\\ 
%					\qquad$\canModel,w\models\knows{a}{(M[\vec{M}/\vec{m}])}$ and 
%						$\phi'[\vec{M}/\vec{m}]\in w$\hfill $\knows{a}{(M[\vec{M}/\vec{m}])}\in\mathcal{P}$\\[\jot]
%				\quad iff there is $\vec{M}\in\vec{\messages}$ \st\\ 
%					\qquad$\canModel,w\models\knows{a}{(M[\vec{M}/\vec{m}])}$ and 
%						$\canModel,w\models\phi'[\vec{M}/\vec{m}]$\hfill by the ind.\ hyp.\\[\jot]
%				\quad iff there is $\vec{M}\in\vec{\messages}$ \st
%					$\canModel,w\models\knows{a}{(M[\vec{M}/\vec{m}])}\land\phi'[\vec{M}/\vec{m}]$\hfill by def.\\[\jot]
%				\quad iff $\canModel,w\models\exists\vec{m}(\knows{a}{(M[\vec{m}])} \land \phi'[\vec{m}])$\hfill by def.
%			\end{flushleft}
	\end{enumerate}

	With the Truth Lemma we can now prove that 
			for all $\phi\in\pFormulas$,
				if $\not\LiiPded\phi$ then $\not\models\phi$.
	Let 
		$\phi\in\pFormulas$, and 
	suppose that 
		$\not\LiiPded\phi$.
	Thus, 
		$\set{\neg\phi}$ 
			is LiiP-consistent, and 
			can be extended to a maximally LiiP-consistent set $w$, \ie
				$\neg\phi\in w\in\mathcal{W}$.
	Hence 
		$\canModel,w\models\neg\phi$, by the Truth Lemma.
	Thus: 
		$\canModel,w\not\models\phi$,
		$\canModel\not\models\phi$, and
		$\not\models\phi$.
	That is,
			$\canModel$ is a 
				\emph{universal} (for \emph{all} $\phi\in\pFormulas$) 
				\emph{counter-model} (if $\phi$ is a non-theorem then $\canModel$ falsifies $\phi$).
	
	We are left to prove that 
		$\canModel$ is also an \emph{LiiP-model}. 
	So let us instantiate our data mining operator $\clo{a}{}$ (\cf Page~\pageref{page:DataMining}) on $\mathcal{W}$  
		by letting for all $w\in\mathcal{W}$
			$$\msgs{a}(w)\defeq\setst{M}{\knows{a}{M}\in w},$$ and let us prove that:
					\begin{enumerate}
						\item if $\community\subseteq\community'$
								then ${\canrel{M}{a}{\community}}\subseteq{\canrel{M}{a}{\community'}}$
						\item if $M\equiv_{a}M'$ then ${\canrel{M}{a}{\community}}={\canrel{M'}{a}{\community}}$
						\item if $M\in\clo{a}{w}(\emptyset)$ 
								then $w\canrel{M}{a}{\community}w$
						\item if $w\canrel{\sign{M}{b}}{a}{\community}w'$ then $M\in\clo{b}{w'}(\emptyset)$
						\item for all $b\in\community\cup\set{a}$, 
							$({\canrel{\sign{M}{a}}{b}{\community\cup\set{a}}}\circ{\canrel{M}{a}{\community}})\subseteq
								{\canrel{\pair{M}{b}}{a}{\community}}$.
					\end{enumerate} 
	For (1),
		let $\community'\subseteq\agents$ and
		suppose that $\community\subseteq\community'$.
	That is,
		$\community\cup\community'=\community'$.
	Further,
		let $w,w'\in\mathcal{W}$ and 
		suppose that $w\canrel{M}{a}{\community}w'$.
	That is,
		for all $\phi\in\pFormulas$,
			if $\proves{M}{\phi}{a}{\community}\in w$ 
			then $\phi\in w'$.
	Furthermore,	
		let $\phi\in\pFormulas$ and
		suppose that $\proves{M}{\phi}{a}{\community'}\in w$.
	Thus $\proves{M}{\phi}{a}{\community\cup\community'}\in w$ by the first supposition.
	Since $w$ is maximal,
		$$\text{$(\proves{M}{\phi}{a}{\community\cup\community'})\limp\proves{M}{\phi}{a}{\community}\in w$\quad(group decomposition).}$$
	Hence $\proves{M}{\phi}{a}{\community}\in w$ by \emph{modus ponens}, and 
	thus $\phi\in w'$ by the second supposition.
	
	For (2),
		suppose that 
			$M\equiv_{a}M'$.
	That is,
		for all $w\in\mathcal{W}$, 
			$M\in\clo{a}{w}(\emptyset)$ if and only if $M'\in\clo{a}{w}(\emptyset)$.
	Hence for all $w\in\mathcal{W}$,
		$\knows{a}{M}\in w$ if and only if $\knows{a}{M'}\in w$
				due to the maximality of $w'$, which 
					contains all the term axioms corresponding to the defining clauses of $\clo{a}{w}$.
	Hence for all $w\in\mathcal{W}$,
			$\canModel, w\models\knows{a}{M}$ if and only if $\canModel, w\models\knows{a}{M'}$, by the Truth Lemma.
	Thus for all $w\in\mathcal{W}$, $\canModel, w\models\knows{a}{M}\lequiv\knows{a}{M'}$.
	Hence for all $w\in\mathcal{W}$, 
		$\knows{a}{M}\lequiv\knows{a}{M'}\in w$ by the Truth Lemma.
	Hence the following intermediate result, called IR, 
			$$\text{for all $w\in\mathcal{W}$ and $\phi\in\pFormulas$, 
			$(\proves{M}{\phi}{a}{\community})\lequiv\proves{M'}{\phi}{a}{\community}\in w$,}$$
		by 
			epistemic bitonicity.
	Further,
		let $w,w'\in\mathcal{W}$.
	Hence,
		\begin{itemize}
			\item $w\canrel{M}{a}{\community}w'$ by definition if and only if
			\item (for all $\phi\in\pFormulas$,
					if $\proves{M}{\phi}{a}{\community}\in w$ 
					then $\phi\in w'$) by IR if and only if 
			\item (for all $\phi\in\pFormulas$,
					if $\proves{M'}{\phi}{a}{\community}\in w$ 
					then $\phi\in w'$) by definition if and only if 
			\item $w\canrel{M'}{a}{\community}w'$.
		\end{itemize}	
	
	For (3), 
		let $w\in\mathcal{W}$ and 
		suppose that 
			$M\in\clo{a}{w}(\emptyset)$.
	Hence $\knows{a}{M}\in w$ due to the maximality of $w$, which 
		contains all the term axioms corresponding to the defining clauses of $\clo{a}{w}$.
	Further suppose that $\proves{M}{\phi}{a}{\community}\in w$.
	Since $w$ is maximal,
		$$\text{$(\proves{M}{\phi}{a}{\community})\limp(\knows{a}{M}\limp\phi)\in w$\quad(epistemic truthfulness).}$$
	Hence, 
		$\knows{a}{M}\limp\phi\in w$, and
		$\phi\in w$, by consecutive  \emph{modus ponens.}
	
	For (4),
		let $w,w'\in\mathcal{W}$ and 
		suppose that $w\canrel{\sign{M}{b}}{a}{\community}w'$.
	That is,	
		for all $\phi\in\pFormulas$, 
			if $\proves{\sign{M}{b}}{\phi}{a}{\community}\in w$
			then $\phi\in w'$.
	Since $w$ is maximal,
		$$\text{$\proves{\sign{M}{b}}{\knows{b}{M}}{a}{\community\cup\set{b}}\in w$\quad(authentic knowledge)}$$
	and 
		$$\text{$(\proves{\sign{M}{b}}{\knows{b}{M}}{a}{\community\cup\set{b}})\limp\proves{\sign{M}{b}}{\knows{b}{M}}{a}{\community}\in w$\quad(group decomposition).}$$
	Hence, 
		$\proves{\sign{M}{b}}{\knows{b}{M}}{a}{\community}\in w$ by \emph{modus ponens}, 
		$\knows{b}{M}\in w'$ by particularisation of the supposition, and thus 
		$M\in\clo{b}{w'}(\emptyset)$ by the definition of $\clo{b}{w'}$.
	
	For (5), 
		suppose that $b\in\community\cup\set{a}$ and 
		let $w,w',w''\in\states$.
	Further suppose that 
			$w\canrel{\sign{M}{a}}{b}{\community\cup\set{a}}w'$ 
				(\ie for all $\phi\in\pFormulas$,
					if $\proves{\sign{M}{a}}{\phi}{b}{\community\cup\set{a}}\in w$ 
					then $\phi\in w'$) and 
			$w'\canrel{M}{a}{\community}w''$
				(\ie for all $\phi\in\pFormulas$,
					if $\proves{M}{\phi}{a}{\community}\in w'$ 
					then $\phi\in w''$).
	Furthermore suppose that
		$\proves{\pair{M}{b}}{\phi}{a}{\community}\in w$.
	Since $w$ is maximal, 
		$$\text{$(\proves{\pair{M}{b}}{\phi}{a}{\community})\limp\proves{\sign{M}{a}}{(\proves{M}{\phi}{a}{\community})}{b}{\community\cup\set{a}}\in w$,}$$
	as a direct consequence of nominal peer review and then the first supposition.
	Hence, applying \emph{modus ponens} consecutively, 
		$\proves{\sign{M}{a}}{(\proves{M}{\phi}{a}{\community})}{b}{\community\cup\set{a}}\in w$ by the fourth supposition, 
		$\proves{M}{\phi}{a}{\community}\in w'$ by particularisation of the second supposition, and finally 
		$\phi\in w''$ by the third supposition.
\end{proof}

\section{Application examples}\label{section:ApplicationExamples}
With the simple but powerful language of LiiP, 
	we can concisely express 
		otherwise difficult to formalise security requirements such as those arising in 
			\emph{Access Control} (\cf \cite[Chapter~4]{SecurityEngineering}) and 
			\emph{Data-Base Privacy} (\cf \cite[Chapter~9]{SecurityEngineering}).
\subsection{Access Control}
According to \cite[Chapter~4]{SecurityEngineering}: 
	\begin{quotation}
		Access control is the traditional center of gravity of computer security.
		It is where security engineering meets computer science. Its function is
		to control which principals (persons, processes, machines\ldots) have
		access to which resources in the system --- which files they can read,
		which programs they can execute, how they share data with other principals, and so on.
	\end{quotation}
``Principals'' and ``resources'' mean ``agents'' in our terminology. 
Access rights can be specified by application-specific access-control \emph{policies} $\Phi$; and
specific access is then granted when certain access-authorisation \emph{credentials} $C$ are presented.
These credentials are examples of application-specific base data $B$ 
	(\cf Definition~\ref{definition:LiiPLanguage}), whose
		validity typically is,  
			first, temporary and thus non-monotonic 
				as in the case of 
					one-time credentials and 
					credentials revokable by other, so-called revocation credentials, and,  
			second, restricted to certain agent communities $\community\subseteq\agents$.
Conceptually, 
	an access-control policy can be understood as a set $\Phi$ of 
		implicational laws $\phi$ that 
		together with elementary access-right facts $P$ constitutes 
			a Horn-logical (\cf Prolog) or even an efficiently decidable Datalog theory.
In LiiP, 
	we can formalise each elementary access-right fact as 
		an application-specific atomic proposition $P\in\mathcal{P}$.
An example of such a fact is that 
	an agent $a$ may write-access resource $r$ guarded by a different agent $b$ 
		(acting thus as a \emph{reference monitor}), which
			we can formalise as an atomic proposition 
				$P_{1}\defeq\mathsf{maywrite}(a,r,b)$.
Thus we can let $\community\defeq\set{a,b}\subset\set{a,r,b}\subseteq\agents$.
Naturally,
	agent $a$ may then also read resource $r$ guarded by agent $b$, \ie
		$\phi_{1}\defeq(\mathsf{maywrite}(a,r,b)\limp\mathsf{mayread}(a,r,b))$.
\emph{Et cetera} up to $\phi_{m}$ and $P_{n}$ for some natural numbers $m,n\in\mathbb{N}$.
Now, define the resulting access-control policy as 
	$\Phi\defeq\set{\phi_{i}}_{1\leq i\leq m}\,,$
the resulting access-control LiiP-theory over $\Phi$ as 
	$$\LiiP_{\Phi}\defeq\Clo{}{}(\Phi)$$
(where $\Clo{}{}$ is as in Definition~\ref{definition:AxiomsRules}), and
	$\vdash_{\LiiP_{\Phi}}$ similarly to $\LiiPded$.
Whence the following instance of a direct consequence of nominal peer review
$$\vdash_{\LiiP_{\Phi}}(\proves{\pair{C}{b}}{\mathsf{maywrite}(a,r,b)}{a}{\community})\limp
	\proves{\sign{C}{a}}{\mathsf{maywrite}(a,r,b)}{b}{\community}$$
and the following instance of epistemic truthfulness
$$\vdash_{\LiiP_{\Phi}}(\proves{\sign{C}{a}}{\mathsf{maywrite}(a,r,b)}{b}{\community})\limp
	(\knows{b}{\sign{C}{a}}\limp\mathsf{maywrite}(a,r,b))\,.$$	
Hence by transitivity of logical implication,
$$\vdash_{\LiiP_{\Phi}}(\proves{\pair{C}{b}}{\mathsf{maywrite}(a,r,b)}{a}{\community})\limp
	(\knows{b}{\sign{C}{a}}\limp\mathsf{maywrite}(a,r,b))\,.$$	
This means that 
	if it is commonly accepted in $\community$ that 
		$\pair{C}{b}$ can prove to (and thus inform) $a$ that $a$ may write-access $r$ guarded by $b$,
	then if further $b$ knows $\sign{C}{a}$ 
		(through $a$ presenting $\sign{C}{a}$ to $b$, since
			only $a$ can generate her own signature),
	then indeed $a$ may write-access $r$---and the guard $b$ knows that (due to Fact~\ref{fact:KC}) and
		thus will grant $a$ the requested access.
Actually $b$ will also grant $a$ read-access since according to the policy $\Phi$, write access implies read access:
$$\vdash_{\LiiP_{\Phi}}\begin{array}[t]{@{}l@{}}
	(\proves{\pair{C}{b}}{\mathsf{maywrite}(a,r,b)}{a}{\community})\limp(\knows{b}{\sign{C}{a}}\limp\\
	\qquad(\mathsf{maywrite}(a,r,b)\land\mathsf{mayread}(a,r,b)))\,.
\end{array}$$
Note that 
	we could 
		refine our arguably rough policy $\Phi$ with respect to \emph{agent roles} and thus
		specify a refined policy $\Phi'$.
For example, 
	we could specify that 
		$\Phi\subseteq\Phi'$ and 
		that for all $x,y\in\community$, 
			$\mathsf{guest}(x)\in\mathcal{P}$ and $\mathsf{host}(y)\in\mathcal{P}$ as well as 
			$((\mathsf{guest}(x)\land\mathsf{host}(y))\limp\mathsf{mayread}(x,r,y))\in\Phi'$, 
			$(\mathsf{host}(y)\limp\mathsf{maywrite}(y,r,y))\in\Phi'$.
\emph{Et cetera.} 
Orthogonally to agent roles, 
	we could refine $\Phi$ with respect to 
		\emph{agent clearances} and corresponding 
		\emph{resource classifications} (\cf \emph{Information Flow Control} \cite[Section~8.3.1--2]{SecurityEngineering}) and
		thus specify a refined policy $\Phi''$.
For example we could specify that 
	$\Phi\subseteq\Phi''$ and that 
	for all $a\in\agents$ (and thus for all resources $r$),
		$\mathsf{topsecret}(a),
			\mathsf{secret}(a),
			\mathsf{confidential}(a),
			\mathsf{unclassified}(a)\in\mathcal{P}$ as well as
		$(\mathsf{topsecret}(a)\limp\mathsf{secret}(a))\in\Phi''$,
		$(\mathsf{secret}(a)\limp\mathsf{confidential}(a))\in\Phi''$,  
		$((\mathsf{topsecret}(a)\newline\land\mathsf{topsecret}(r))\limp\mathsf{maywrite}(a,r,b))\in\Phi''$, and
		$((\mathsf{topsecret}(a)\land
			(\mathsf{secret}(r)\lor\linebreak\mathsf{confidential}(r)\lor\mathsf{unclassified}(r)))\limp
				\mathsf{mayread}(a,r,b))\in\Phi''$.
\emph{Et cetera} for other, so-called \emph{no-read-up} and \emph{no-write-down/up} requirements.

\subsection{Data-Base Privacy}
An important example of a resource is a \emph{relational data-base,} say a medical data-base $d$, 
	over application-specific atomic pieces of content data $B$
		(\cf Definition~\ref{definition:LiiPLanguage}).
Note that $d$ typically evolves, whence the point of non-monotonicity.
Then, each unary relation in the data-base $d$ can be understood as a finite (sub)set of content data $B$,
each binary relation as a finite set of ordered pairs $\pair{B}{B'}$ of data $B$ and $B'$, 
each relation of higher finite arity as a finite set of such pairs of pairs, and 
the content of $d$ as a finite set of such relations (finite sets).
Finally, finite sets can be coded as data pairs and thus
the entire content $\mathcal{D}$ of $d$ can be understood as a subset of $\messages$ over 
	the atomic data $B$.
Now, data-base privacy with respect to the data-base $d$ means that
	certain agents $a$ must not be able to infer certain facts $\phi$ from $d$
		(\cf \emph{Inference Control} \cite[Section~9.3]{SecurityEngineering}).
In order to meet this privacy requirement,
	certain atomic data $B$ in $\mathcal{D}$ are blinded (\eg replaced by some dummy datum),
		resulting in a new, partially blinded content $\mathcal{D}'\subseteq\messages'$.
The privacy requirement can now be formalised in the language of LiiP by 
	simply stipulating that for all $M\in\messages'$, 
		$$\neg(\proves{M}{\phi}{a}{\emptyset})\,.$$
The requirement could be proved by 
	induction over the well-structured data.
%
% find D' such that privacy holds and the (Hamming-)distance between D and D' is minimal

%\subsection{Trust management}
%correctness proof as trust inducers
\end{document}